\documentclass[12pt]{article}
\usepackage[colorlinks,linkcolor=blue,citecolor=blue,urlcolor=blue,bookmarks,bookmarksnumbered]{hyperref}
\usepackage{mathpazo,charter}
\usepackage[scaled=0.95]{helvet}
\usepackage{amsmath,XXXE}
\usepackage{graphicx,color}
\usepackage{booktabs,multirow,array}

\def\rD{{\rm D}}

\def\hk{{\hat{k}}}

\def\cK{\mathcal{K}}

\def\sL{\mathscr{L}}
\def\Dt#1{\accentset{\hbox{\large.}}{#1}}
\def\DDt#1{\accentset{\hbox{\large.\kern-1pt.}}{#1}}
\def\dt#1{\accentset{\hbox{\normalsize.}}{#1}}
\def\ddt#1{\accentset{\hbox{\normalsize\kern.5pt.\kern-1pt.}}{#1}}
\def\vdt{\partial_\tau}
\def\ivt#1{{\buildrel{\makebox[0pt]{\,.\kern-1pt.\kern-1pt.\kern-1pt.}}\over{#1}}}

\definecolor{Green}  {rgb}{0.10,0.70,0.10} %  1
\definecolor{Orange} {rgb}{1.00,0.50,0.15} %  2
\definecolor{Red}    {rgb}{0.90,0.00,0.12} %  3
\definecolor{Purple} {rgb}{0.42,0.15,0.45} %  4
\definecolor{Turque} {rgb}{0.00,0.65,0.85} %  5
\definecolor{Blue}   {rgb}{0.00,0.00,1.00} %  6
\definecolor{Magenta}{rgb}{1.00,0.00,1.00} %  7
\definecolor{Gold}   {rgb}{1.00,0.75,0.25} %  8
\definecolor{Seaweed}{rgb}{0.01,0.24,0.09} %  9
\definecolor{Brown}  {rgb}{0.43,0.26,0.32} % 10
\definecolor{grey1}  {rgb}{0.20,0.20,0.20} % 11
\definecolor{grey2}  {rgb}{0.40,0.40,0.40} % 12
\definecolor{grey3}  {rgb}{0.60,0.60,0.60} % 13
\definecolor{grey4}  {rgb}{0.80,0.80,0.80} % 14
\definecolor{grey5}  {rgb}{0.90,0.90,0.90} % 15
\def\C#1#2{{\ifcase#1\or
             \color{Green}\or \color{Orange}\or \color{Red}\or
              \color{Purple}\or \color{Turque}\or \color{Blue}\or
               \color{Magenta}\or \color{Gold}\or \color{Seaweed}\or
                \color{Brown}\or\color{grey1}\or\color{grey2}\or
                 \color{grey3}\else\color{grey4}\fi#2}}
\def\cb#1#2{\setlength\fboxsep{1pt}\colorbox{#1}{\color{#1}\fbox{\color{black}#2}}}
\def\cB#1{\hbox to0pt{\setlength\fboxsep{0pt}\hss\color{grey3}\fbox{\cb{white}{#1}}\hss}}
\def\bB#1{\hbox to0pt{\setlength\fboxsep{0pt}\hss\color{grey3}\fbox{\cb{black}{\color{white}#1}}\hss}}
\def\vdt{\vd_\tau}

 \SfTitles
 \allowdisplaybreaks

\begin{document}

\thispagestyle{empty}
 \noindent
 \today\hfill
  \vspace*{5mm}
 \begin{center}
{\LARGE\sf\bfseries Golden Ratio Controlled Chaos\\[2mm]
                     in Supersymmetric Dynamics}\\[1mm]
  \vspace*{5mm}
 \begin{tabular}{p{80mm}cp{80mm}}
 \hfill
{\large\sf\bfseries  Tristan H\"{u}bsch$^{*\dag}$}
                    &and&
{\large\sf\bfseries  Gregory A.~Katona$^{\dag\ddag}$}\\[1mm]
\MC3c{\small\it
  $^*$\,Department of Physics \&\ Astronomy,
  Howard University, Washington, DC 20059} \\[-1mm]
\MC3c{\small\it
  $^\dag$\,Department of Physics, University
  of Central Florida, Orlando, FL 32816}\\[-1mm]
\MC3c{\small\it
  $^\ddag$\,Affine Connections, LLC, College Park, MD 20740}\\[0mm]
 \hfill {\tt  thubsch@howard.edu}&&{\tt grgktn@knights.ucf.edu}
  \end{tabular}\\[1mm]
  \vspace*{5mm}
{\sf\bfseries ABSTRACT}\\[2mm]
\parbox{148mm}{We construct supersymmetric Lagrangians for the recently constructed off-shell worldline $N\,{=}\,3$ supermultiplet $\IY_I/(i\rD_I\IX)$ for $I=1,2,3$, where $\IY_I$ and $\IX$ are standard, Salam-Strathdee superfields: $\IY_I$ fermionic and $\IX$ bosonic. Already the Lagrangian bilinear in component fields exhibits a total of thirteen free parameters, seven of which specify Zeeman-like coupling to external (magnetic) fluxes. All but special subsets of this parameter space describe aperiodic oscillatory response, some of which are controlled by the ``golden ratio,'' $\vf\approx1.61803$. We also show that all of these Lagrangians admit an $N\,{=}\,3\to4$ supersymmetry extension, while a subset admits two inequivalent such extensions.
 } 
\end{center}
\vspace{5mm}
\noindent
\parbox[t]{60mm}{PACS: {\tt11.30.Pb}, {\tt12.60.Jv}}\hfill
\parbox[t]{100mm}{\raggedleft\small\baselineskip=12pt\sl
             The elevator to success is out of order.\\[-1pt]
             You'll have to use the stairs, one step at a time.\\[-0pt]
            |~Joe Girard}
\vspace{5mm}

\section{Introduction, Results and Synopsis}
\label{s:IRS}
Systematic construction and exploration of off-shell supermultiplets has received a significant boost from focusing on dimensional reduction to worldline $N$-extended supersymmetry (see \eg, Refs.\cite{rGR1,rGR2,rPT,rGLPR,rGLP,rA,rKRT,r6-1,rKT07,r6--1,r6-1.2,r6-3.1,rTHGK12,rGIKT12}). The dimensional extension of these to higher-dimensional spacetimes is being actively investigated\cite{rFIL,rFL,rGH-obs,rH-WWS,rJP-DBt}, and has already produced new off-shell supermultiplets in higher-dimensional spacetime, such as finite-dimensional off-shell realizations of the Fayet ``hypermultiplet''\cite{r6-4}, which in turn is crucial in studies of $T$-duality in superstring and similar theories\cite{rMMYau1,rMMYau2,rMMYau3}. 
 In addition,
 the underlying theoretical framework for $M$-theory includes worldline supersymmetry in a prominent way.
 Finally, the Hilbert space of every supersymmetric field theory necessarily admits the action of an induced worldline supersymmetry.
 Consistent quantum formulation of any such theory requires off-shell fields for path integration, and this provides the fundamental motivation for exploring the structure of off-shell representations of $N$-extended worldline supersymmetry.

Resolving a puzzle from Ref.\cite{r6-1}, we have proven\cite{rTHGK12} that a semi-infinite iterative sequence of off-shell supermultiplet quotients produces infinitely many ever larger indecomposable off-shell supermultiplets of worldline $N\,{\geqslant}\,3$ supersymmetry. Herein, we focus on the smallest of these novel, indecomposable off-shell supermultiplets, for $N\,{=}\,3$ supersymmetry and with $5{+}8{+}3$ component fields $(\f_i|\j_\hk|F_m)$, and the supersymmetry transformation rules specified (adapting the notation from Ref.\cite{rTHGK12}) as
\begin{equation}
  \begin{array}{@{} c|c@{~~}c@{~~}c @{}}
 & \C3{Q_1} & \C1{Q_2} & \C6{Q_3} \\ 
    \hline\rule{0pt}{2.1ex}
\f_1
 & \j_1 & \j_2 & \j_3 \\[0pt]
\f_2
 & \j_3 & -\j_4 & -\j_1 \\[0pt]
\f_3
 & \j_4\C5{{-}\j_7} & \j_3\C5{{-}\j_5}
   & -\j_2\C5{{+}\j_6} \\[-2pt]
\f_4
 & \j_5 & -\j_7 & -\j_8 \\[0pt]
\f_5
 & -\j_6 & \j_8 & -\j_7 \\*[3pt]
    \cline{2-4}\rule{0pt}{2.5ex}
F_1
 & \dt\j_2 & -\dt\j_1 & \dt\j_4 \\[0pt]
F_2
 & \dt\j_8 & \dt\j_6 & \dt\j_5 \\[0pt]
F_3
 & \dt\j_7 & \dt\j_5 & -\dt\j_6 \\
    \hline
  \end{array}
 \qquad
  \begin{array}{@{} c|c@{~~}c@{~~}c @{}}
 & \C3{Q_1} & \C1{Q_2} & \C6{Q_3} \\ 
    \hline\rule{0pt}{2.9ex}
\j_1
 & i\dt\f_1 & -iF_1 & -i\dt\f_2 \\[0pt]
\j_2
 & iF_1 & i\dt\f_1 & -i\dt\f_3\C5{{-}iF_3} \\[0pt]
\j_3
 & i\dt\f_2 & i\dt\f_3\C5{{+}iF_3} & i\dt\f_1 \\[0pt]
\j_4
 & i\dt\f_3\C5{{+}iF_3} & -i\dt\f_2 & iF_1 \\[0pt]
\j_5
 & i\dt\f_4 & iF_3 & iF_2 \\[0pt]
\j_6
 & -i\dt\f_5 & iF_2 & -iF_3 \\[0pt]
\j_7
 & iF_3 & -i\dt\f_4 & -i\dt\f_5 \\[0pt]
\j_8
 & iF_2 & i\dt\f_5 & -i\dt\f_4 \\[0pt]
    \hline
  \end{array}
  \label{e:GKY}
\end{equation}
which may be depicted by the graph in Figure~\ref{f:A>B}: component fields are depicted by like-labeled edges and the $Q$-transformation between them are depicted by edges, solid (dashed) for positive (negative) signs in\eq{e:GKY}. Edge colors distinguish between
 $\C3{Q_1}$, $\C1{Q_2}$ and $\C6{Q_3}$, and the tapered edges indicate a one-way $Q$-action; \eg, $\C3{Q_1}(\j_4)\supset F_3$ but $\C3{Q_1}(F_3)\not\supset\j_4$.
\begin{figure}[ht]
\centering
 \begin{picture}(140,45)
   \put(0,0){\includegraphics[width=140mm]{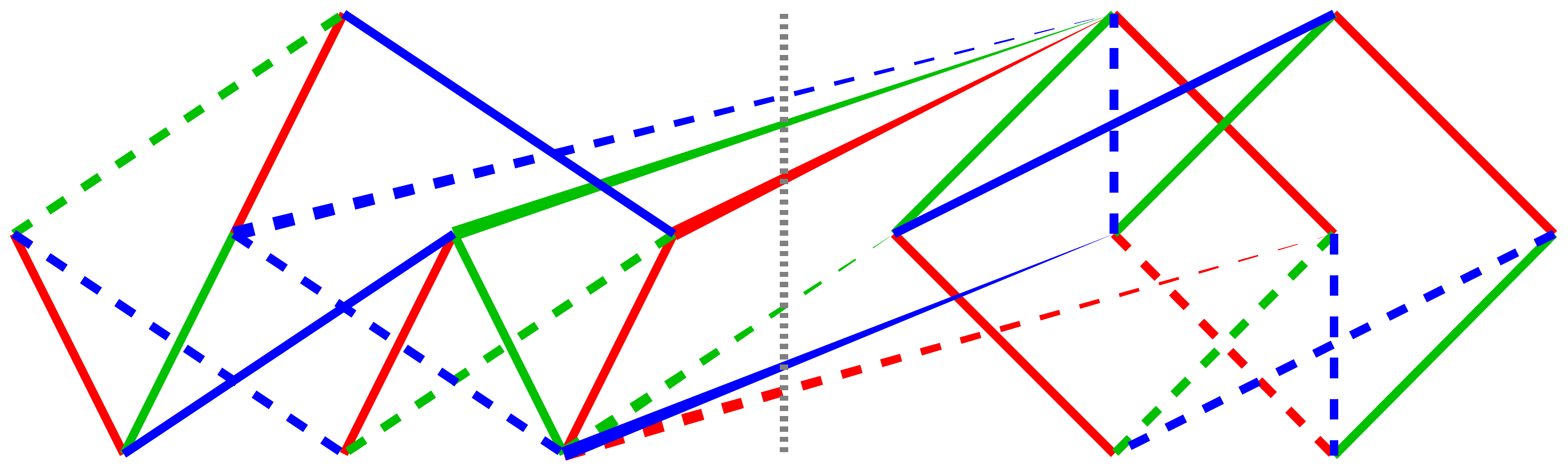}}
    \put(12,1){\cB{$\f_1$}}
    \put(31,1){\cB{$\f_2$}}
    \put(50.5,1){\cB{$\f_3$}}
    \put(99.5,1){\cB{$\f_4$}}
    \put(119,1){\cB{$\f_5$}}
    \put(2,20){\bB{$\j_1$}}
    \put(21,20){\bB{$\j_2$}}
    \put(40,20){\bB{$\j_3$}}
    \put(59,20){\bB{$\j_4$}}
    \put(80,20){\bB{$\j_5$}}
    \put(99,20){\bB{$\j_6$}}
    \put(118,20){\bB{$\j_7$}}
    \put(137,20){\bB{$\j_8$}}
    \put(31,40){\cB{$F_1$}}
    \put(99.5,40){\cB{$F_3$}}
    \put(119,40){\cB{$F_2$}}
 \end{picture}
\caption{A graphical depiction of the gauge-quotient supermultiplet of Ref.\protect\cite{rTHGK12}.}
 \label{f:A>B}
\end{figure}
The supermultiplet\eq{e:GKY} exhibits no obvious conventional symmetry nor does it support a complex structure, which helps identifying features that are dictated by supersymmetry itself, unencumbered by consequences of any other symmetry or structure. Ref.\cite{rTHGK12} showed that this supermultiplet cannot be {\em\/decomposed\/} into a direct sum of smaller supermultiplets, although it may be {\em\/reduced\/} by setting one half of it to zero---but only the fields in the right-hand half of Figure~\ref{f:A>B}. This type of {\em\/indecomposable reducibility\/}\ft{This is radically different from representations of Lie algebras, where reduction implies decomposition as a direct sum. Infinitely many off-shell supermultiplets of which\eq{e:GKY} is but the simplest example may be reduced by constraining some portion to zero, but do not decompose as a direct sum. Such representations are therefore definitely ``more than the sum of their parts.''} turns out to be ubiquitous in the indefinite sequence of ever-larger supermultiplets of Ref.\cite{rTHGK12}, and persists to higher-dimensional spacetimes\cite{r6-4.2}.

Throughout this article, we focus on the worldline $N$-extended supersymmetry algebra with no central charge is
\begin{equation}
 \left.
 \begin{aligned}
 \big\{\, Q_I \,,\, Q_J \,\big\}&=2\d_{IJ}\,H,&
 \big[\, H \,,\, Q_I \,\big] &=0,\\
 Q_I^{~\dagger}&= Q_I,& H^\dag&=H,\quad
\end{aligned}\right\}\quad
 I,J=1,2,\cdots,N,
  \label{e:SuSyQ}
\end{equation}
were $H\,{=}\,i\vdt$. In Section~\ref{s:SSL}, we present a 13-parameter family of bilinear supersymmetric Lagrangians for the supermultiplet\eq{e:GKY}. A 7-parameter subset describes Zeeman-like coupling to external magnetic fluxes, which includes regimes of chaotic response.
Section~\ref{s:MSI} generalizes this to an infinite-dimensional family of interactive (non-bilinear) Lagrangians, all of which are supersymmetric by explicit construction. Section~\ref{s:SuSyX} then explores enhancing worldline supersymmetry from $N\,{=}\,3\to4$ and its effects on the possible choices of a Lagrangian. Our closing comments are collected in Section~\ref{s:Coda}; technical details and derivations are deferred to the appendix~\ref{s:Bits}.

\section{Manifestly Supersymmetric Lagrangians}
\label{s:SSL}
This section presents a class of Lagrangians for the supermultiplet\eq{e:GKY} which are supersymmetric by construction, akin to the ones obtained by standard superfield methods\cite{r1001,rPW,rWB,rBK}; see in particular Ref.\cite{rFGH}.

\subsection{Manifestly Supersymmetric Kinetic Terms}
\label{s:KE}
Consider the quantity
\begin{equation}
  \sL_{\sss\vec{A}}^{\sss\text{KE}} \Defl \C6{Q_3}\C1{Q_2}\C3{Q_1}\,A^{i\hk}\f_i\j_\hk
 \label{e:L0}
\end{equation}
with $A^{i\hk}$ an arbitrary $5\,{\times}\,8=40$ matrix of coefficients. It is manifestly supersymmetric regardless of the choice of the coefficients $A^{i\hk}$, since $Q_I\sL_{\sss\vec{A}}^{\sss\text{KE}}$ is a total time-derivative for each $I=1,2,3$:
\begin{equation}
  Q_I\big[\C6{Q_3}\C1{Q_2}\C3{Q_1}\,A^{i\hk}\f_i\j_\hk\big]
  =\vdt\Big[(-1)^{I+1}i\prod_{J\neq I}Q_J\,A^{i\hk}\f_i\j_\hk\Big].
 \label{e:QL0}
\end{equation}
Being true for any choice of $A^{i\hk}$, this construction seems to provide a 40-parameter family of manifestly supersymmetric Lagrangian terms, analogous to the so-called ``$D$-terms'' within the standard constructions with simple supersymmetry in 4-dimensional spacetime\cite{r1001,rPW,rWB,rBK}.

 The Lagrangian terms\eq{e:L0} are dimensionally adequate for kinetic terms in a Lagrangian if we choose the standard mass-dimension for the component fields:
\begin{equation}
  [\f_i]      =-\frc12\quad\To\quad
  [\j_\hk]    =0\quad\text{and}\quad
  [F_m]       =+\frc12,
 \label{e:dim}
\end{equation}
and choose the $A^{i\hk}$ to be dimensionless numerical constants. Then,
\begin{equation}
  [\sL_{\sss\vec{A}}^{\sss\text{KE}}]
  =[\C6{Q_3}\C1{Q_2}\C3{Q_1}\,A^{i\hk}\f_i\j_\hk] = +1,
  \quad\text{and}\quad
  [\int\rd\t\,\sL_{\sss\vec{A}}^{\sss\text{KE}}] = 0.
\end{equation}
However, these forty terms are not all unrelated, and Table~\ref{t:KE} specifies six linearly independent manifestly supersymmetric kinetic Lagrangian terms.
\begin{table}[htpb]
\vspace{-1mm}
$$
  \begin{array}{@{} r@{\,\Defl}r|r@{\,\Defl\,}l @{}}
\MC2{c|}{\bs{k^a(\f,\j)}}
 &\MC2l{\text{\bsf Supersymmetric Kinetic Terms
         \boldmath$K^a(\f,\j,F)\Defl\C6{Q_3}\C1{Q_2}\C3{Q_1}\,k^a(\f,\j)$}}
 \\*[1pt]
    \toprule
k^1&\f_1\j_4 & K^1 &
 F_1^{~2} +(F_3{+}\dt\f_3)^2 +\dt\f_1^{~2} +\dt\f_2^{~2}
  +i\j_1\dt\j_1 +i\j_2\dt\j_2 +i\j_3\dt\j_3 +i\j_4\dt\j_4
  \\*
k^2&\f_4 \j_6 & K^2 &
 F_2^{~2} + F_3^{~2} + \dt\f_4^{~2} + \dt\f_5^{~2}
 +i\j_5\dt\j_5 +i\j_6\dt\j_6 +i\j_7\dt\j_7 +i\j_8\dt\j_8
 \\*
k^3&\f_1\j_5 & K^3 & 
   F_1 F_2 - F_3 \dt\f_2 + (F_3 + \dt\f_3)\dt\f_4 + \dt\f_1 \dt\f_5
 - i\j_1 \dt\j_6 + i\j_2 \dt\j_8 - i\j_3 \dt\j_7 + i\j_4 \dt\j_5 
 \\*
k^4&-\f_1\j_6 & K^4 &
   F_1 F_3 + F_2 \dt\f_2 + (F_3 + \dt\f_3) \dt\f_5 - \dt\f_1 \dt\f_4
 - i\j_1 \dt\j_5 + i\j_2 \dt\j_7 + i\j_3 \dt\j_8 - i\j_4 \dt\j_6 
 \\*
k^5&\f_1\j_7 & K^5 & 
   F_3(F_3+\dt\f_3) - F_1 \dt\f_5 + F_2 \dt\f_1 + \dt\f_2 \dt\f_4 
 + i\j_1 \dt\j_8 + i\j_2 \dt\j_6 + i\j_3 \dt\j_5 + i\j_4 \dt\j_7 
 \\*
k^6&\f_1\j_8 & K^6 &
   F_2(F_3+\dt\f_3) - F_1 \dt\f_4 - F_3 \dt\f_1 - \dt\f_2 \dt\f_5 
 - i\j_1 \dt\j_7 - i\j_2 \dt\j_5 + i\j_3 \dt\j_6 + i\j_4 \dt\j_8
\\ 
    \bottomrule
  \end{array}
$$\vspace{-7mm}
  \caption{Linearly independent manifestly supersymmetric kinetic terms}
  \label{t:KE}
\end{table}
These were taken from the full listing, given in Table~\ref{t:KE0}, in the appendix. Already in that table, some of the entries vanish, indicating that the corresponding $\C6{Q_3}\C1{Q_2}\C3{Q_1}(\f_i\j_\hk)$ term is a total time-derivative. Further identities reduce this list to the six rows of Table~\ref{t:KE}; see appendix~\ref{a:KE} for further details.

This leaves the manifestly supersymmetric kinetic Lagrangian\eq{e:L0} expressible as
\begin{subequations}
 \label{e:KE}
\begin{align}
  \sL_{\sss\vec{A}}^{\sss\text{KE}}
 &\Defl \sum_{a=1}^6 A_a\,K^a(\f,\j,F)
   = \inv2\C6{Q_3}\C1{Q_2}\C3{Q_1}\Big(\sum_{a=1}^6 A_a\,k^a(\f,\j)\Big),\\
 &~=\inv2\C6{Q_3}\C1{Q_2}\C3{Q_1}\big[  A_1\f_1\j_4
                                     {+}A_2\f_4\j_6
                                      + A_3\f_1\j_5
                                     {-}A_4\f_1\j_6
                                     {+}A_5\f_1\j_7
                                     {+}A_6\f_1\j_8 \big]
\end{align}
\end{subequations}
and depending explicitly on the array $\vec{A}=(A_1,\cdots,A_6)$ of six free parameters.

For example, the choice
\begin{equation}
 \begin{aligned}
  \sL_{\sss(1,0,0,0,0,0)}^{\sss\text{KE}}
  &=\inv2\C6{Q_3}\C1{Q_2}\C3{Q_1}(\f_1\j_4),\\
  &=\inv2\big( F_1^{~2}
              +\dt\f_1^{~2}
              +\dt\f_2^{~2}
              +(F_3{+}\dt\f_3)^2 \big)
    +\frc{i}2\big(\j_1\dt\j_1
              + \j_2\dt\j_2 
              +\j_3\dt\j_3
              +\j_4\dt\j_4 \big)
 \end{aligned}
 \label{e:KE=L}
\end{equation}
is---except for the $\dt\f_3{+}F_3$ mixing---the standard kinetic term for the left-hand half of the supermultiplet, as it is depicted in Figure~\ref{f:A>B}. This mixing with a component field from the right-hand half of the supermultiplet, $\dt\f_3$ with $F_3$, owes to the existence of the one-way supersymmetry action depicted in Figure~\ref{f:A>B} by the tapering edges---of which $F_3$ is the sole target component. We note in passing that the equation of motion of $F_3$ is then $F_3=-\dt\f_3$, the use of which completely eliminates both $F_3$ and $\dt\f_3$ from this simple Lagrangian.

 The nonlocal field redefinition $\Tw\f_3\Defl\f_3+\int\rd\t\,F_3$ would remove this mixing and would decompose the supermultiplet\eq{e:GKY} into a direct sum of the left-hand half and the right-hand half of the graph in Figure~\ref{f:A>B}. However, a nonlocal field redefinition is not acceptable as an equivalence relation, the supermultiplet\eq{e:GKY} does not decompose\cite{rTHGK12}, and the kinetic terms\eq{e:KE=L} remain mixed with $F_3$ from the other half of the supermultiplet.
 
 In turn, the choice
\begin{equation}
 \begin{aligned}
  \sL_{\sss(0,1,0,0,0,0)}^{\sss\text{KE}}
  &=\inv2\C6{Q_3}\C1{Q_2}\C3{Q_1}(\f_4\j_6),\\
  &=\inv2\big( F_2^{~2}
              +F_3^{~2}
              +\dt\f_4^{~2}
              +\dt\f_5^{~2} \big)
   +\frc{i}2\big(\j_5\dt\j_5
              +\j_6\dt\j_6
              +\j_7\dt\j_7 
              +\j_8\dt\j_8 \big)
 \end{aligned}
 \label{e:KE=R}
\end{equation}
provides the standard kinetic term for the right-hand half of the supermultiplet as depicted in Figure~\ref{f:A>B}. This portion involves none of the ``left-hand side'' components, $(\f_1,\f_2,\f_3|\j_1,\j_2,\j_3,\j_4|F_1)$, since none of these are the target components of any of the one-way supersymmetry action; see Figure~\ref{f:A>B}.

This identifies the simple sum of\eq{e:KE=L} and\eq{e:KE=R}:
\begin{equation}
 \begin{aligned}
  \sL_{(1,1,0,0,0,0)}^{\sss\text{KE}}
  &=\frc12\C6{Q_3}\C1{Q_2}\C3{Q_1}(\f_1\j_4+\f_4\j_6),\\
  &=\frc12\big[ \dt\f_1^{~2}
               +\dt\f_2^{~2}
               +(F_3{+}\dt\f_3)^2
               +\dt\f_4^{~2}
               +\dt\f_5^{~2} \big]
   +\frc12\big[ F_1^{~2}
               +F_2^{~2}
               +F_3^{~2}\big]
\\
  &\quad
   +\frc{i}2\big[\j_1\dt\j_1
                +\j_2\dt\j_2
                +\j_3\dt\j_3
                +\j_4\dt\j_4 
                +\j_5\dt\j_5 
                +\j_6\dt\j_6
                +\j_7\dt\j_7 
                +\j_8\dt\j_8 \big],
 \end{aligned}
 \label{e:StdKE}
\end{equation}
as the familiar-looking standard kinetic term in a supersymmetric Lagrangian for the supermultiplet\eq{e:GKY}. The equation of motion of $F_3$ from\eq{e:StdKE} is $F_3=-\fc12\dt\f_3$. Eliminating $F_3$ from the Lagrangian\eq{e:StdKE} changes the $\fc12\dt\f_3^{~2}$-term into $\fc14\dt\f_3^{~2}$ and so reduces the effective ``mass'' of $\f_3$ from $1$ to $\inv{\sqrt2}$, but does not eliminate this propagating component field from the Lagrangian---unlike the case with the ``left-hand side'' terms\eq{e:KE=L} alone.

The 6-dimensional complement $\{\vec{A}\in\IR^6,~\vec{A}\neq(1,1,0,0,0,0)\}$ then para\-met\-ri\-zes ``non-standard'' but manifestly supersymmetric kinetic terms. In retrospect, the fact that the kinetic Lagrangian terms for the well-known and well-used models are unique\cite{r1001,rBK} stem not only from a higher supersymmetry\ft{Simple (${\cal N}\,{=}\,1$) supersymmetry in 4-dimensional spacetime, \eg, is equivalent to $N\,{=}\,4$ on the worldline.} (see also Section~\ref{s:SuSyX}), but also from conventional symmetries (including the Lorentz group) and/or additional (complex and/or hyper-complex) structures.

The full 6-parameter dependence of the kinetic terms\eq{e:KE} however does remain available in all worldline and some worldsheet applications, as could be useful in string theory and its $M$- and $F$-theory extensions: The graphical rendition in Figure~\ref{f:A>B} easily shows that, as per the ``bow-tie'' theorem of Ref.\cite{rGH-obs}, the supermultiplet\eq{e:GKY} extends to a {\em\/worldsheet\/} off-shell supermultiplet of $(3,0)$- or $(0,3)$-supersymmetry.

The standard kinetic Lagrangian\eq{e:StdKE} has all the summands in a uniform format and with a positive sign, allowing for a straightforward construction of a partition functional and the corresponding unitary quantum model.
 Each of the $K^3,\cdots,K^6$ terms mixes the component fields from the left-hand side and the right-hand side of the supermultiplet as depicted in Figure~\ref{f:A>B}. This necessarily opens the possibility of non-unitarity, which becomes manifest upon diagonalizing the component fields to ``normal modes.'' For example, the Lagrangian
\begin{alignat}9
 \sL_{(1,1,\a,0,0,0)}^{\sss\text{KE}}
 &=\ldots +i\j_1\dt\j_1+\ldots +i\j_6\dt\j_6+\ldots
   +\frc{i}2\a(\j_1\dt\j_6-\dt\j_1\j_6)+\ldots
 \label{e:11a}
\end{alignat}
has normal modes that are eigenvectors of the matrix $\big[\begin{smallmatrix}2&\a\\\a&2\end{smallmatrix}\big]$, and which appear in the diagonalized Lagrangian with coefficients ${\propto}\,(2{\pm}\a)$. For both of these to have the positive sign, we must limit the range of this free parameter to $|\a|<2$. For the Lagrangian\eq{e:11a}, this choice suffices to insure the positivity of all fermionic kinetic energy terms, as well as the diagonalized $(\f_1,\f_5)$- and $(F_1,F_2)$-terms. The diagonalization of $(\f_2,\f_3,\f_4,F_3)$-terms is more involved, but yields the simple reduction of the unitarity range to $|\a|<\fc1{\sqrt2}$.
 The unitarity conditions on the full 6-dimensional parameter space of\eq{e:KE} are of course considerably more involved. Nevertheless, the parameter space does include a continuum of unitary (positive) Lagrangians\eq{e:11a} when $|\a|<\fc1{\sqrt2}$.

In principle, any particular choice from among the Lagrangians\eq{e:KE} may be transformed into a ``standard-looking'' kinetic Lagrangian
\begin{equation}
 \frc12\sum_{i=1}^5\dt\vf_i{}^2
 +\frc{i}2\sum_{\hk=1}^8\c_\hk\dt\c_\hk
  +\frc12\sum_{m=1}^3 G_m{}^2
 \label{e:StdKE0}
\end{equation}
by means of a judicious field redefinition. However, all but very special---and incomplete, such as\eq{e:KE=R}---choices of the $\vec{A}$-dependent Lagrangians will require nonlocal field redefinitions such as $\Tw\f_3\Defl\f_3+\int\rd\t\,F_3$, to transform into the form\eq{e:StdKE0}. Every such field redefinitions also significantly complicates the supersymmetry transformation rules, being that\eq{e:GKY} is the simplest representative of the gauged quotient continuum of choices\cite{rTHGK12}.
 Suffice it here to note that we may take\eq{e:StdKE} as the starting point, where the $\dt\f_3+F_3$ mixing cannot be removed by local field redefinition. The complementary 6-parameter variation\eq{e:KE} then describes nontrivial variations in dynamics, although a precise determination of which of these variations remain inequivalent under local field redefinitions is beyond our present scope.

\subsection{Manifestly Supersymmetric Zeeman Terms}
\label{s:ZE}
The Lagrangian\eq{e:KE} consists of bilinear terms with the mass-dimension $+1$, its $\t$-integral has the physical units of $\hbar$, and so requires no multiplicative mass-parameter. The constants $\vec{A}$ in the Lagrangian\eq{e:KE} must be chosen to be purely numerical.

It turns out, however, that the supermultiplet\eq{e:GKY} may also have supersymmetric Lagrangian terms with mass-dimension 0, thus requiring at least one mass-like multiplicative parameter. Such terms were also found in Ref.\cite{r6-7a}, where they arose as Zeeman-like interactions of a supermultiplet with external magnetic fields/fluxes.

To this end and for the off-shell supermultiplet\eq{e:GKY}, we have:
\begin{lemm}\label{T:Z}
Let $z(\f)\Defl z^{ik}\f_i\f_k$ be an expression that is bilinear in the lowest components of the supermultiplet\eq{e:GKY}, chosen such that
\begin{equation}
  \C6{Q_3}\C1{Q_2}\C3{Q_1}\big[z(\f_k)\big]\simeq 0~\text{(mod $\vdt\cK$)},\qquad
 \cK=\cK(\f,\j,F)~\text{analytic}.
\end{equation}
 The expressions
\begin{equation}
  \C3{Q_1}\C1{Q_2}\big[z(\f)\big],\quad
  \C3{Q_1}\C6{Q_3}\big[z(\f)\big]\quad \text{and}\quad
  \C1{Q_2}\C6{Q_3}\big[z(\f)\big]
 \label{e:ZE}
\end{equation}
are then manifestly supersymmetric, and linear combinations of these expressions that are not total time-derivatives are nontrivial supersymmetric Lagrangians.
\end{lemm}
\begin{proof}
The general supersymmetry transformation of the first of the expressions\eq{e:ZE} is generated by
\begin{align}
 \e^IQ_I\,\C3{Q_1}\C1{Q_2}\big[z(\f)\big]
 &=i\e^1\big[\C3{Q_1^{~2}}\C1{Q_2}z(\f)\big]
   -i\e^2\big[\C3{Q_1}\C1{Q_2^{~2}}z(\f)\big]
  -\e^3\C6{Q_3}\C1{Q_2}\C3{Q_1}\big[z(\f)\big],\nn\\
 &=\vdt\big[i\e^1\C1{Q_2}z(\f)-i\e^2\C3{Q_1}z(\f)+\e^3\cK\big].
 \label{e:ZEP}
\end{align}
For the remaining two expressions\eq{e:ZE}, the roles of the three terms in the expansion\eq{e:ZEP} change cyclically but produce a total time-derivative just the same.
\end{proof}
This construction provides Lagrangian terms analogous to the so-called ``$F$-terms'' within standard constructions with simple supersymmetry in 4-dimensional spacetime\cite{r1001,rPW,rWB,rBK}.

Table~\ref{t:ZE0} in the appendix lists the $\binom{5+1}2=15$ expressions of the form $\C6{Q_3}\C1{Q_2}\C3{Q_1}(\f_i\f_k)$. Besides the two vanishing entries therein, straightforward row-operations find several additional vanishing bilinear combinations of this sort. Thus,
\begin{equation}
 z(\f)=\bigg\{
 \begin{array}{c}
  \fc12(\f_1^{~2}{-}\f_2^{~2}),~~
  (\f_1\f_2),~~
  (\f_1\f_3{+}\f_2\f_5),~~
  (\f_1\f_4{+}\f_3\f_5),~~
  (\f_1\f_5{-}\f_2\f_3),\\
  (\f_1\f_5{-}\f_3\f_4),~~
  \fc12(\f_4^{~2}-\f_5^{~2}),~~
  (\f_4\f_5),~~
  \fc12(\f_1^{~2}{-}\f_3^{~2}{+}\f_5^{~2}{-}2\f_2\f_4)
\end{array}
 \label{e:pZb}
\end{equation}
provide a basis of nine linearly independent bilinear functions for the construction in Lemma~\ref{T:Z}. From each of these, we compute the three possible $Q_IQ_J$-transforms\eq{e:ZE}, as listed in Tables~\ref{t:ZE01}--\ref{t:ZE03} in appendix~\ref{a:ZE}. Within these, row-operations straightforwardly reduce to the final list of linearly independent supersymmetric super-Zeeman Lagrangian terms, listed in Table~\ref{t:SZ} by simplicity.
\begin{table}[htpb]
\vspace{-1mm}
$$
  \begin{array}{@{} r@{\,\Defl\,}l @{}}
\MC2c{\text{\bsf\boldmath The Super-Zeeman Lagrangian Terms $iQ_IQ_J\,z(\f)$}}
 \\[1pt]
  \toprule
Z^1&
  i\j_5\j_7
 -i\j_6\j_8
 -\f_4F_3
 -\f_5F_2
\\
Z^2&
  i\j_5\j_8
 +i\j_6\j_7
 -\f_4F_2
 +\f_5F_3
\\[1pt]
Z^3&
  i\j_1\j_2
 +i\j_3\j_4
 -\f_1F_1
 -\f_2F_3
 -\f_2\dt\f_3
\\
Z^4&
  i\j_1\j_4 
 +i\j_2\j_3
 -\f_1F_3
 +\f_2F_1
 -\f_1\dt\f_3
\\[1pt]
Z^5&
  i\j_1\j_8
 +i\j_2\j_6
 +i\j_3\j_5
 +i\j_4\j_7
 -\f_1F_2
 -\f_3F_3
 -\f_5F_1
 -\f_2\dt\f_4 
\\[1pt]
Z^6&
  i\j_1\j_7
 +i\j_2\j_5
 -i(\j_3{-}\j_5)\j_6
 -i(\j_4{-}\j_7)\j_8
 -\f_1F_3
 +\f_3F_2
 +\f_4F_1
 -\f_2\dt\f_5
 +\f_4\dt\f_5 
\\[1pt]
Z^7&
  i\j_1(\j_3{-}\j_5)
 -i\j_2(\j_4{-}\j_7)
 +i\j_3\j_8
 -i\j_4\j_6
 -\f_2F_2
 -\f_3F_1
 +\f_5F_3
 -\f_1\dt\f_2
 +\f_1\dt\f_4
 -\f_3\dt\f_5 
\\
\bottomrule
  \end{array}
$$\vspace{-7mm}
  \caption{The linearly independent supersymmetric super-Zeeman Lagrangian terms. All seven appear for every choice of $I,J=1,2,3$, with $z(\f)$ from the collection\eq{e:pZb}; see Tables~\ref{t:ZE01}--\ref{t:ZE03}.}
  \label{t:SZ}
\end{table}

Except for the first two entries, $Z^1,Z^2$, the terms in Table~\ref{t:SZ} contain expressions of the form $\f_j\dt\f_k\simeq\inv2\big(\f_j\dt\f_k-\f_k\dt\f_j\big)$, up to total derivatives. With\eq{e:StdKE} chosen as the kinetic term, $\dt\f_j$ is proportional to the momentum canonically conjugate to $\f_j$, and $\big(\f_j\dt\f_k{-}\f_k\dt\f_j\big)$ is proportional to the angular momentum corresponding to rotations in the $(\f_j,\f_k)$-plane. The remainder of summands in the row then provide the supersymmetric completion of this angular momentum. A coefficient in the Lagrangian multiplying such a term is then the external magnetic flux coupling to this angular momentum. The latter five rows of Table~\ref{t:SZ} indicate that there are five such independent fluxes coupling supersymmetrically to the five corresponding combinations of angular momenta in the 5-dimensional $(\f_2,\f_4,\f_3,\f_1,\f_5)$-space.

The first two rows of Table~\ref{t:SZ} do not involve the $(\f_2,\f_4,\f_3,\f_1,\f_5)$-bosons, and so cannot be interpreted as such a magnetic flux coupling. Nevertheless, they do belong into this sector of the Lagrangian on dimensional grounds, and we refer to them also as super-Zeeman terms. We thus write
\begin{equation}
 \sL_{\sss\vec{B}}^{\sss\text{SZ}}\Defl B_a Z^a(\f,\j,F),
 \label{e:SZ}
\end{equation}
where $\vec{B}=(B_1,\cdots,B_7)$ is a seven-component array of background fluxes, and $Z^a(\f,\j,F)$ are the seven expressions from Table~\ref{t:SZ}.

\subsection{Simple Harmonic Helicoidal Response}
Consider for example, the term
\begin{equation}
  \sL_{\sss(0,0,B_3,0,0,0,0)}^{\sss\text{SZ}}
 = B_3\big[i\j_1 \j_2 + i \j_3 \j_4
           - F_1 \f_1 - F_3 \f_2 - \fc12(\f_2\dt\f_3-\f_3\dt\f_2)\big].
 \label{e:SZ3}
\end{equation}
It describes a super-Zeeman interaction of the supermultiplet\eq{e:GKY} with a magnetic flux $B_3$ passing through the $(\f_2,\f_3)$-plane. Adding\eq{e:SZ3} to the lagrangian\eq{e:StdKE}, the equations of motion for $F_1$ and $F_3$ set
\begin{equation}
  \d_{F_1}\sL=0:~~F_1\to B_3\f_1\quad\text{and}\quad
  \d_{F_3}\sL=0:~~F_3\to\fc12(B_3\f_2-\dt\f_3),
\end{equation}
which eliminates $F_1$ and $F_3$ from the Lagrangian, modifying it (upon integration by parts) into
\begin{equation}
 \begin{aligned}
  &\sL_{\sss(1,1,0,0,0,0)}^{\sss\text{KE}}+\sL_{\sss(0,0,B_3,0,0,0,0)}^{\sss\text{SZ}}\\
  &\qquad
    \to+\fc12\dt\f_1^{~2}+\fc12\dt\f_2^{~2}+\fc14\dt\f_3^{~2}
        +\fc12\dt\f_4^{~2}+\fc12\dt\f_5^{~2}
         -\fc12B_3^{~2}(\f_1^{~2}+\fc12\f_2^{~2})+\fc12 B_3\dt\f_2\,\f_3+\dots
 \end{aligned}
 \label{e:LSZ3}
\end{equation}
The resulting system of Euler-Langrange equations of motion are
\begin{subequations}
 \label{e:Z3}
\begin{alignat}9
 \ddt\f_1+B_3^{~2}\f_1
 &=0\qquad&\text{so}\qquad
 \f_1&=a_1\sin\big(B_3\t+\d_1\big),\\*
 \ddt\f_2 +\fc12B_3^{~2}\f_2 +\fc12B_3\dt\f_3
 &=0\qquad&\text{so}\qquad
 \f_2&=a_2\sin\big(B_3\t+\d_2\big)+\D_2,\\*
 \ddt\f_3 -B_3\dt\f_2
 &=0\qquad&\text{so}\qquad
 \f_3&=-a_2\cos\big(B_3\t+\d_2\big)-\D_2 B_3\t+\D_3,\\*
 \ddt\f_4
 &=0\qquad&\text{so}\qquad
 \f_4&=v_4\t+\D_4,\\*
 \ddt\f_5
 &=0\qquad&\text{so}\qquad
 \f_5&=v_5\t+\D_5,
\end{alignat}
\end{subequations}
where the amplitudes $a_1,a_2$, the phases $\d_1,\d_2$, the velocities $v_4,v_5$ and the displacements $\D_2,\D_3,\D_4$ and $\D_5$ are the integration constants parametrizing the simple harmonic, helicoidal motion\eq{e:Z3}, all with one and the same frequency, equal to $B_3$.

It is evident from Table~\ref{t:SZ} that including the super-Zeeman Lagrangian term $Z^4(\f)$ instead of $Z^3(\f)$ in\eq{e:SZ3} would have a similar effect: the $(\f_1,\f_2;\j_2,\j_4)\to(-\f_2,\f_1;\j_4,-\j_2)$ swap ``rotates'' $(Z^3,Z^4)\to(Z^4,-Z^3)$. Indeed, adding any one of $Z^3,Z^4,Z^5$ to the standard kinetic Lagrangian\eq{e:StdKE}, eliminating the auxiliary fields $F_1,F_2$ and $F_3$ by means of their equations of motion and then solving the equations of motion for $\f_1,\cdots,\f_5$ results in a helicoidal response akin to\eq{e:Z3}.

In fact, even the first two terms, $Z^1$ and $Z^2$, induce a similar helicoidal response although they do not include $(\f_i\dt\f_k\,{-}\,\f_k\dt\f_i)$-like angular momentum terms. For example, adding $B_1 Z^1$ to the standard kinetic Lagrangian\eq{e:StdKE} and eliminating $F_1,F_2$ and $F_3$ through their equations of motion produces the Lagrangian
\begin{equation}
 \begin{aligned}
  &\sL_{\sss(1,1,0,0,0,0)}^{\sss\text{KE}}+\sL_{\sss(B_1,0,0,0,0,0,0)}^{\sss\text{SZ}}\\
  &\qquad
    \to+\fc12\dt\f_1^{~2} +\fc12\dt\f_2^{~2} +\fc14\dt\f_3^{~2}
        +\fc12\dt\f_4^{~2} +\fc12\dt\f_5^{~2}
         -\fc12B_1^{~2}(\fc12\f_4^{~2}+\f_5^{~2})+\fc12 B_1\dt\f_3\,\f_4+\dots
 \end{aligned}
 \label{e:LSZ1}
\end{equation}
which equals\eq{e:LSZ3} upon the $(\f_1,\f_2,\f_3,\f_4,\f_5)\to(\f_5,\f_4,-\f_3,\f_2,\f_1)$ field redefinition and integration by parts on the last term. Thus, despite different appearances and the lack of $(\f_i\dt\f_k\,{-}\,\f_k\dt\f_i)$-like angular momentum terms in $Z^1$ and $Z^2$, the Lagrangian terms in Table~\ref{t:SZ} indeed all couple the supermultiplet\eq{e:GKY} to external magnetic fluxes, thus justifying the grouping in Table~\ref{t:SZ}. Of these, each one of the first five induces this type of simple harmonic, helicoidal response.

\subsection{Golden Ratio Chaos}
Adding either $Z^6$ or $Z^7$ to the standard kinetic Lagrangian\eq{e:StdKE} however results in a radically different response. Consider the bilinear $N\,{=}\,3$-supersymmetric worldline Lagrangian
\begin{subequations}
 \label{e:L6}
\begin{alignat}9
 \sL
  &=\sL^{\sss\text{KE}}_{\sss(1,1,0,0,0,0)} + \sL^{\sss\text{SZ}}_{\sss(0,0,0,0,0,B_6,0)},\\
  &=\frc12\big[ \dt\f{}_1^{~2}
               +\dt\f{}_2^{~2}
               +(F_3{+}\dt\f_3)^2
               +\dt\f{}_4^{~2}
               +\dt\f{}_5^{~2} \big]
   +\frc12\big[ F_1^{~2}
               +F_2^{~2}
               +F_3^{~2}\big] +\dots
\\
  &\quad
  +B_6\big[{-}\f_1F_3
           +\f_3F_2
           +\f_4F_1
           -\f_2\dt\f_5
           +\f_4\dt\f_5 +\dots\big],
\end{alignat}
\end{subequations}
where the fermion-fermion terms have been omitted.
 The equations of motion for the auxiliary fields $F_1,F_2,F_3$ are:
\begin{alignat}9
 \d_{F_1}:
 &&\qquad
 0
 &\Is F_1+B_6\f_4
 \quad&\To\quad
 F_1&=-B_6\f_4;\\
 \d_{F_2}:
 &&\qquad
 0
 &\Is F_2+B_6\f_3
 \quad&\To\quad
 F_2&=-B_6\f_3;\\
 \d_{F_3}:
 &&\qquad
 0
 &\Is (F_3+\dt\f_3)+F_3-B_6\f_1
 \quad&\To\quad
 F_3&=\frc12(B_6\f_1-\dt\f_3).
\end{alignat}
Substituting these back into\eq{e:L6}, we obtain:
\begin{subequations}
 \label{e:L6f}
\begin{alignat}9
 \sL
  &=\frc12 \dt\f{}_1^{~2}
   +\frc14\dt\f{}_3^{~2}
   +\frc12B_6\f_1\dt\f_3
   -\frc14 B_6^{~2} \f_1^{~2}
   -\frc12 B_6^{~2} \f_3^{~2}
\label{e:L6a}\\
  &\quad
   +\frc12 \dt\f{}_2^{~2}
   +\frc12 \dt\f{}_4^{~2}
   +\frc12 \dt\f{}_5^{~2}
   -B_6\f^{}_2\dt\f_5
   +B_6\f^{}_4\dt\f_5
   -\frc12 B_6^{~2} \f_4^{~2}
   +\dots
\label{e:L6b}
\end{alignat}
\end{subequations}
Notice that the five bosonic fields decouple into two groups: the $(\f_1,\f_3)$-plane and the $(\f_2,\f_4,\f_5)$-volume.
 The two resulting linear differential systems of bosonic equations of motion are:
\begin{subequations}
 \label{e:f13}
\begin{alignat}9
0&= \frc{\rd}{\rd t}\big(\dt\f_1\big)
    -\big(\inv2B_6\dt\f_3-\inv2B_6^{~2}\f_1\big)&\qquad
0&=2\ddt\f_1-B_6\dt\f_3+B_6^{~2}\f_1, \label{e:f13a}\\*
\smash{\raisebox{4.5mm}{$\left\{\rule{0pt}{8mm}\right.$}}
0&= \frc{\rd}{\rd t}\Big(\inv2\dt\f_3+\inv2B_6\f_1\big)
    -\big({-}B_6^{~2}\f_3\big)
 &\qquad\smash{\raisebox{4.5mm}{$\left\{\rule{0pt}{8mm}\right.$}}
0&=\ddt\f_3+B_6\dt\f_1+2B_6^{~2}\f_3;\label{e:f13b}
\end{alignat}
\end{subequations}
and
\begin{subequations}
 \label{e:f245}
\begin{alignat}9
0&= \frc{\rd}{\rd t}\big(\dt\f_2\big)
    -\big({-}B_6\dt\f_5\big)&\qquad
0&=\ddt\f_2+B_6\dt\f_5,\label{e:f245a}\\*
\smash{\raisebox{.5mm}{$\left\{\rule{0pt}{12mm}\right.$}}
0&= \frc{\rd}{\rd t}\big(\dt\f_4\big)
    -\big(B_6\dt\f_5-B_6^{~2}\f_4\big)&\qquad
\smash{\raisebox{.5mm}{$\left\{\rule{0pt}{12mm}\right.$}}
0&=\ddt\f_4-B_6\dt\f_5+B_6^{~2}\f_4,\label{e:f245b}\\*
0&= \frc{\rd}{\rd t}\big(\dt\f_5 -B_6\f^{}_2 +B_6\f^{}_4\big)
    -\big(0\big)&\qquad
0&=\ddt\f_5-B_6\dt\f_2+B_6\dt\f_4.\label{e:f245c}
\end{alignat}
\end{subequations}

Consider the planar system\eq{e:f13} first. Substituting $\dt\f_3=(2\ddt\f_1+B_6^{~2}\f_1)/B_6$ from\eq{e:f13a} into the derivative of\eq{e:f13b}, we obtain
\begin{equation}
 \frc2{B_6}(\ivt\f_1+3B_6^{~2}\ddt\f_1+B_6^{~4}\f_1) = 0.
 \label{e:ddddf1}
\end{equation}
Looking for a solution in the form $\f_1=A_1\sin(\w\t{+}\d_1)$, we obtain:
\begin{equation}
  (\w^4-3B_6^{~2}\w^2+B_6^{~4})A_1\sin(\w\t{+}\d_1)=0,
 \label{e:Ansatz}
\end{equation}
the (positive) solutions of which
\begin{equation}
  \w_+=\vf B_6,\quad \w_-=\vf^{-1}B_6,
\end{equation}
are irrational multiples of the magnetic flux $B_6$, by factors of the {\em\/Golden Ratio\/},
\begin{equation}
  \vf\Defl\frc{\sqrt5+1}2\approx1.61803,
   \qquad\text{for which}\qquad
  \vf^{-1}=\vf{-}1.
 \label{e:GR}
\end{equation}
The general $(\f_1,\f_3)$-solution may thus be written as
\begin{subequations}
 \label{e:f1f3}
\begin{alignat}9
 \f_1
 &=a_1\sin(\vf B_6\t{+}\d_1)
   +\tw{a}_1\sin(\vf^{-1} B_6\t{+}\tw\d_1), \label{e:f1}
\intertext{and, substituting $\ddt\f_3$ from\eq{e:f13a} into\eq{e:f13b}:}
 \f_3
 &=-\frc2{B_6^{~3}}\dddt\f_1 -\frc1{B_6}\dt\f_1,\label{e:f1>f3}\\
 &=a_1\vf^2\cos(\vf B_6\t{+}\d_1)
   -\tw{a}_1\vf^{-2}\cos(\vf^{-1} B_6\t{+}\tw\d_1),\label{e:f3}
\end{alignat}
\end{subequations}
owing to the identity\eq{e:GR} so $\vf^3{-}\vf=\vf^2$ and $\vf^{-3}{-}\vf^{-1}=-\vf^{-2}$.
 The solutions\eq{e:f1f3} are parametrized by four integration constants, the amplitudes $a_1,\tw{a}_1$ and the phases $\d_1,\tw\d_1$.

Turning now to the $(\f_2,\f_4,\f_5)$-volume, we express $\dt\f_5=-\ddt\f_2/B_6$ from\eq{e:f245a} and substitute this into\eqs{e:f245b}{e:f245c}:
\begin{equation}
 \begin{array}{r@{~=~}llr@{~=~}l}
  0& \ddt\f_4 +\ddt\f_2 +B_6^{~2}\f_4
   &\To&\ddt\f_2&-\ddt\f_4-B_6^{~2}\f_4,\\[1mm]
  0& -\frc1{B_6}\dddt\f_2 -B_6\dt\f_2 +B_6\dt\f_4
   &\To&0&\ivt\f_2 +B_6^{~2}\ddt\f_2 -B_6^{~2}\ddt\f_4.\\
 \end{array}
\end{equation}
Substituting now the requisite derivatives of $\f_2$ from the top equation into the bottom one produces
\begin{equation}
  \ivt\f_4 +3B_6^{~2}\ddt\f_4 +B_6^{~4}\f_4 = 0,
\end{equation}
which is identical to\eq{e:ddddf1} and so is analogously solved by
\begin{subequations}
 \label{e:f4f2f5}
\begin{equation}
  \f_4 =a_4\sin(\vf B_6\t{+}\d_4)
        +\tw{a}_4\sin(\vf^{-1} B_6\t{+}\tw\d_4).\label{e:f4}
\end{equation}
Substituting $\ddt\f_5$ from the derivative of\eq{e:f245b} into\eq{e:f245c}, we compute:
\begin{alignat}9
 \f_2 &= 2\f_4+B_6^{~2}\ddt\f_4 + \D_2,\nn\\
      &=-a_4\vf^{-1}\sin(\vf B_6\t{+}\d_4)
        +\tw{a}_4\vf\sin(\vf^{-1} B_6\t{+}\tw\d_4) + \D_2, \label{e:f2}\\
\intertext{where we used that $2{-}\vf^2=-\vf^{-1}$ and $2{-}\vf^{-2}=\vf$. Then, from\eq{e:f245a}:}
 \f_5 &=-\frc1{B_6}\dt\f_2 + \D_5,\nn\\
      &=a_4\cos(\vf B_6\t{+}\d_4)
        -\tw{a}_4\cos(\vf^{-1} B_6\t{+}\tw\d_4) + \D_5.\label{e:f5}
\end{alignat}
\end{subequations}
Note that $\f_2$ and $\f_5$ have undetermined constant displacements, $\D_2$ and $\D_5$, since the system\eq{e:f245} determines only $\dt\f_2$ and $\dt\f_5$. The resulting solutions\eq{e:f4f2f5} thus have the requisite six integration constants: the amplitudes $a_4,\tw{a}_4$, the phases $\d_4,\tw\d_4$ and the displacements $\D_2,\D_5$.

The solutions\eq{e:f1f3} and\eq{e:f4f2f5} clearly indicate that, in response to a {\em\/single\/} magnetic flux, $B_6$, the bosonic component fields $\f_1,\cdots,\f_5$ of the supermultiplet\eq{e:GKY} all oscillate with normal modes of two {\em\/incommensurate\/} frequencies
\begin{equation}
  \w_+ = \vf B_6 = \big(\frc{\sqrt5{+}1}2\big)B_6
   \quad\text{and}\quad
  \w_- = \vf^{-1} B_6 = \big(\frc{\sqrt5{-}1}2\big)B_6,
   \qquad \frac{\w_+}{\w_-}=\vf^2=\frc{\sqrt5{+}3}2,
 \label{e:2freqs}
\end{equation}
distinguished by powers of the {\em\/Golden Ratio\/}\eq{e:GR}. This implies aperiodic (chaotic) oscillatory motion: for any generic (random) choice of initial conditions, the bosonic system $(\f_1,\cdots,\f_5)$ will oscillate, but never return to the initial configuration---although it will come arbitrarily close to it and, as time passes, infinitely many times. After long enough time, the trajectory becomes a telltale toroidal surface filling curve with infinitely many self-intersections.

When coupling any system to at least two independent magnetic fluxes, their relative ratio may be varied continuously and will typically elicit an aperiodic response.
 However, it is highly unusual that component fields of a single, indecomposable supermultiplet such as\eq{e:GKY} respond chaotically to a {\em\/single\/} external magnetic flux to which they couple linearly in a Lagrangian that is bilinear in component fields.
 A more complete exploration of the frequency spectra and regimes as they occur in the various phases of the system with the combined Lagrangian
\begin{equation}
 \sL\Defl\sL_{\sss\vec{A}}^{\sss\text{KE}}+\sL_{\sss\vec{B}}^{\sss\text{SZ}}
 \label{e:KESZ}
\end{equation}
within the 13-dimensional parameter space $(\vec{A};\vec{B})$ is clearly beyond our present scope. The above results however show that there exist at least two radically different phases:
\begin{enumerate}\itemsep=-3pt\vspace{-2mm}
 \item the simple harmonic, helicoidal response regime, as in\eq{e:f1f3},
 \item the chaotic, aperiodic response regime, as in\eq{e:f4f2f5}.
\end{enumerate}

\section{Infinitely Many Manifestly Supersymmetric Interactions}
\label{s:MSI}
One way to introduce non-linearity in the dynamics initially described by the Lagrangian\eq{e:KE} is to make the six parameters $A_a$ into analytic functions of the component fields $\f_1,{\cdots},\f_5$, $\j_1,{\cdots},\j_8$ and $F_1,F_2,F_3$, and then restrict them by requiring the Lagrangian to be supersymmetric.

However, given the mass-dimensions\eq{e:dim} and by a straightforward iteration of the computation\eq{e:QL0}, it is clear that the quantities
\begin{align}
  \sL_{[p,q,r]}^{\sss\text{EG}}(\vec{C})
   &=\C6{Q_3}\C1{Q_2}\C3{Q_1}
      \big[ C^{i_1\cdots i_p\hk_1\cdots\hk_{(2q+1)}m_1\cdots m_r}\,
             \f_{i_1}\cdots\f_{i_p}\,
              \j_{\hk_1}\cdots\j_{\hk_{(2q+1)}}\,
                F_{m_1}\cdots F_{m_r}\big],
\intertext{with $p,r=0,1,2,3,\dots$, and $q=0,1,2,3$,}
 [\sL_{[p,q,r]}^{\sss\text{EG}}(\vec{C})]
   &=3({+}\frc12)+(0)+p({-}\frc12)+(2q{+}1)(0)+r(+\frc12)
    =\frac{3{-}p{+}r}2.
\end{align}
are all supersymmetric, for arbitrary coefficients $C^{^{\,\SSS\dots}}$.
 The quantity $\sL_{[p,q,r]}(\vec{C})$ is bosonic and an analytic function of the fields.
Therefore, with a suitable mass-parameter\ft{Given that the numerical constants $C^{i_1\cdots i_p\hk_1\cdots\hk_{(2q+1)}m_1\cdots m_r}$ are at this stage all arbitrary, it is always possible to fuse all possibly different mass-parameters of a theory into one, the ratios of the different mass-parameters in any desired, particular application being encoded in the choice of the dimensionless parameters $C^{i_1\cdots i_p\hk_1\cdots\hk_{(2q+1)}m_1\cdots m_r}$.} $m$,
\begin{equation}
  \sL_{\vec{C}}^{\sss\text{EG}} \Defl
   \sum_{p,r=0}^\infty m^{\frac{p-r-1}2}
    \sum_{q=0}^3 \sL_{[p,q,r]}^{\sss\text{EG}}(\vec{C})
 \label{e:EG}
\end{equation}
is an infinitely large family of manifestly supersymmetric Lagrangians for the supermultiplet\eq{e:GKY}; since there are only eight fermions, fermionic monomials cannot be of order higher than 8, limiting $q$ as indicated.
 For quantum mechanics in general, and as well for supersymmetric quantum mechanics, concerns of renormalization do not limit the order of the Lagrangian as they do for field theory in higher-dimensional spacetimes. Also, owing to the standard mass-dimension of scalar fields being negative, dimensionless combinations such as $(\f_iF_m)$ and $(m\,\f_i\f_j)$ may occur in the Lagrangian to arbitrary nonnegative powers.
 Although dimensionless, the fermions may appear only in monomials of total degree up to eight owing to their anti-commutativity.

Of these, only the linear combination
\begin{equation}
 \begin{aligned}
   \sL_{[0,0,0]}^{\sss\text{EG}}(\vec{C}) &\Defl \C6{Q_3}\C1{Q_2}\C3{Q_1}(C^\hk \j_\hk)\\
    &~=C^1(\DDt\f_3{+}\Dt{F}_3) -C^2\DDt\f_2 +C^3\Dt{F}_1 -C^4\DDt\f_1
       -C^5\DDt\f_5 -C^6\DDt\f_4 -C^7\Dt{F}_2 +C^8\Dt{F}_3
 \end{aligned}
\end{equation}
is necessarily a trivial total time-derivative.
 The $\sL_{[1,0,0]}$ term in the infinite sum\eq{e:EG} contains the nontrivial $\sL_{\sss\vec{A}}$ in\eq{e:KE}; all other terms parametrize (nonlinear) self-interactions of the supermultiplet $(\f_i|\j_\hk|F_m)$.

For example, the simple choice
\begin{align}
 \C6{Q_3}&\C1{Q_2}\C3{Q_1}(\f_1^{~3}\j_4)\nn\\
 &= -6\j_1\j_2\j_3\j_4
    -6i\f_1\big[F_1(\j_1\j_2{+}\j_3\j_4)
               +\dt\f_2(\j_1\j_3{-}\j_2\j_4)
               +(\dt\f_3{+}F_3)(\j_1\j_4{+}\j_2\j_3)\big]\nn\\
 &~~+3i\f_1^{~2}(\j_1\dt\j_1+\j_2\dt\j_2+\j_3\dt\j_3+\j_4\dt\j_4)
    +3\f_1^{~2}\Big[\dt\f_1^{~2} +\dt\f_2^{~2} +(\dt\f_3{+}F_3)^2 +F_1^{~2}\big]
    \,\C5{-\vdt(\f_1^{~3}\dt\f_1)}
 \label{e:Add4}
\end{align}
already provides fairly nontrivial interactions between the component fields of the left-hand side portion (as depicted in Figure~\ref{f:A>B}) of the supermultiplet\eq{e:GKY}, involving a component field of the right-hand side only through the persistent appearance of the binomial $(\dt\f_3{+}F_3)$. This is the only mixing with the component fields of the right-hand half (as depicted in Figure~\ref{f:A>B}) of the supermultiplet\eq{e:GKY}. Adding an $M$-multiple of\eq{e:Add4} to the standard kinetic terms\eq{e:StdKE} produces the equations of motion for the auxiliary fields
\begin{subequations}
\begin{alignat}9
  \d_{F_1}\sL=0&:&\quad
            F_1&\to 6M\frac{i\f_1(\j_1\j_2{+}\j_3\j_4)}{1+6M\f_1^{~2}},\\
  \text{and}\quad
  \d_{F_3}\sL=0&:&\quad
           F_3 &\to\frac{6M[i\f_1(\j_1\j_2{+}\j_3\j_4)-\f_1\dt\f_3]-\dt\f_3}
                        {2(1+3M\f_1^{~2})}.
\end{alignat}
\end{subequations}
Substituting this back into the Lagrangian\eq{e:StdKE}+$M{\cdot}$(\ref{e:Add4}) and expanding the numerators into geometric series clearly induces highly nontrivial and nonlinear dynamics.

Addionally, the construction of Lemma~\ref{T:Z} that produced $\sL^{\sss\text{SZ}}_{\sss\vec{B}}$ may also be generalized by selecting analytic bosonic ``superpotentials''
\begin{equation}
 \begin{aligned}
   z(\f,\j,F) &\!\Defl Z^{i_1\cdots i_p\,\hk_1\cdots\hk_{2q}\,m_1\cdots m_r}\,
                      \f_{i_1}\cdots\f_{i_p}\,
                       \j_{\hk_1}\cdots\j_{\hk_{2q}}\,
                        F_{m_1}\cdots F_{m_r},\\
 \text{such that}~~
 \C6{Q_3}\C1{Q_2}\C3{Q_1}\,z(\f,\j,F)&\simeq 0~\text{(mod $\vdt\cK$)},\qquad
 \cK=\cK(\f,\j,F)~\text{analytic},
 \end{aligned}
 \label{e:ZEc}
\end{equation}
with $q=0,\cdots4$, but $p,r\in\IN$.
This provides even more supersymmetric linear combinations of interactive terms, spanned by the linearly independent of the expressions
\begin{equation}
  \C3{Q_1}\C1{Q_2}\big[z(\f,\j,F)\big],\quad
  \C3{Q_1}\C6{Q_3}\big[z(\f,\j,F)\big]\quad \text{and}\quad
  \C1{Q_2}\C6{Q_3}\big[z(\f,\j,F)\big].
 \label{e:ZEg}
\end{equation}
We redefine the non-linear Lagrangian\eq{e:EG} so as to include also these terms.

The supersymmetric Lagrangian summands of the form\eq{e:EG} are formally analogous to the so-called ``$D$-terms'' within standard constructions with simple supersymmetry in  4-dimensional spacetime\cite{r1001,rPW,rWB,rBK}. In turn, the supersymmetric Lagrangian terms of the form\eq{e:ZEg} subject to\eq{e:ZEc}, are the formal analogues of the so-called (superpotential) ``$F$-terms''\cite{r1001,rPW,rWB,rBK}: There, the Lagrangian terms are obtained as iterated $Q$-transforms of a superpotential by a subset (one half) of the supercharges, such that the superpotential is annihilated (up to total $\t$-derivatives) by the complementary subset of the supercharges. Analogously, the generalized super-Zeeman Lagrangian terms are here obtained as $Q_IQ_J$-transforms of a ``superpotential,'' requiring that the full (triple) $Q$-transform of this ``superpotential'' is a total $\t$-derivative.

\section{Supersymmetry Extension}
\label{s:SuSyX}
As generally expected of representations of worldline $N\,{=}\,3$ supersymmetry, the supermultiplet\eq{e:GKY} indeed does admit an additional, $4^\text{th}$ supersymmetry---and in at least two distinct ways.

\paragraph{The Chiral-Chiral Extension:}
We simply list the additional supersymmetry transformation rules by extending the tables\eq{e:GKY}:
\begin{equation}
  \begin{array}{@{} c|c@{~~}c@{~~}c@{~~}c @{}}
 & \C3{Q_1} & \C1{Q_2} & \C6{Q_3} & \C8{Q_4} \\ 
    \hline\rule{0pt}{2.1ex}
\f_1
 & \j_1 & \j_2 & \j_3 & \j_4 \\[0pt]
\f_2
 & \j_3 & -\j_4 & -\j_1 & \j_2 \\[0pt]
\f_3
 & \j_4\C5{{-}\j_7} & \j_3\C5{{-}\j_5}
   & -\j_2\C5{{+}\j_6} & -\j_1\C5{{+}\j_8} \\[-2pt]
\f_4
 & \j_5 & -\j_7 & -\j_8 & \j_6 \\[0pt]
\f_5
 & -\j_6 & \j_8 & -\j_7 & \j_5 \\*[3pt]
    \cline{2-5}\rule{0pt}{2.5ex}
F_1
 & \dt\j_2 & -\dt\j_1 & \dt\j_4 & -\dt\j_3 \\[0pt]
F_2
 & \dt\j_8 & \dt\j_6 & \dt\j_5 & \dt\j_7 \\[0pt]
F_3
 & \dt\j_7 & \dt\j_5 & -\dt\j_6 & -\dt\j_8 \\
    \hline
  \end{array}
 \quad~~
  \begin{array}{@{} c|c@{~}c@{~}c@{~}c @{}}
 & \C3{Q_1} & \C1{Q_2} & \C6{Q_3} & \C8{Q_4} \\ 
    \hline\rule{0pt}{2.9ex}
\j_1
 & i\dt\f_1 & -iF_1 & -i\dt\f_2 & -i\dt\f_3\C5{{-}iF_3} \\[0pt]
\j_2
 & iF_1 & i\dt\f_1 & -i\dt\f_3\C5{{-}iF_3} & i\dt\f_2 \\[0pt]
\j_3
 & i\dt\f_2 & i\dt\f_3\C5{{+}iF_3} & i\dt\f_1 & -iF_1 \\[0pt]
\j_4
 & i\dt\f_3\C5{{+}iF_3} & -i\dt\f_2 & iF_1  & i\dt\f_1\\[0pt]
\j_5
 & i\dt\f_4 & iF_3 & iF_2 & i\dt\f_5 \\[0pt]
\j_6
 & -i\dt\f_5 & iF_2 & -iF_3 & i\dt\f_4 \\[0pt]
\j_7
 & iF_3 & -i\dt\f_4 & -i\dt\f_5 & iF_2 \\[0pt]
\j_8
 & iF_2 & i\dt\f_5 & -i\dt\f_4 & -iF_3 \\[0pt]
    \hline
  \end{array}
  \label{e:GKY4}
\end{equation}
which may be depicted in the manner of Figure~\ref{f:A4B}.
\begin{figure}[ht]
\centering
 \begin{picture}(140,45)
   \put(0,0){\includegraphics[width=140mm]{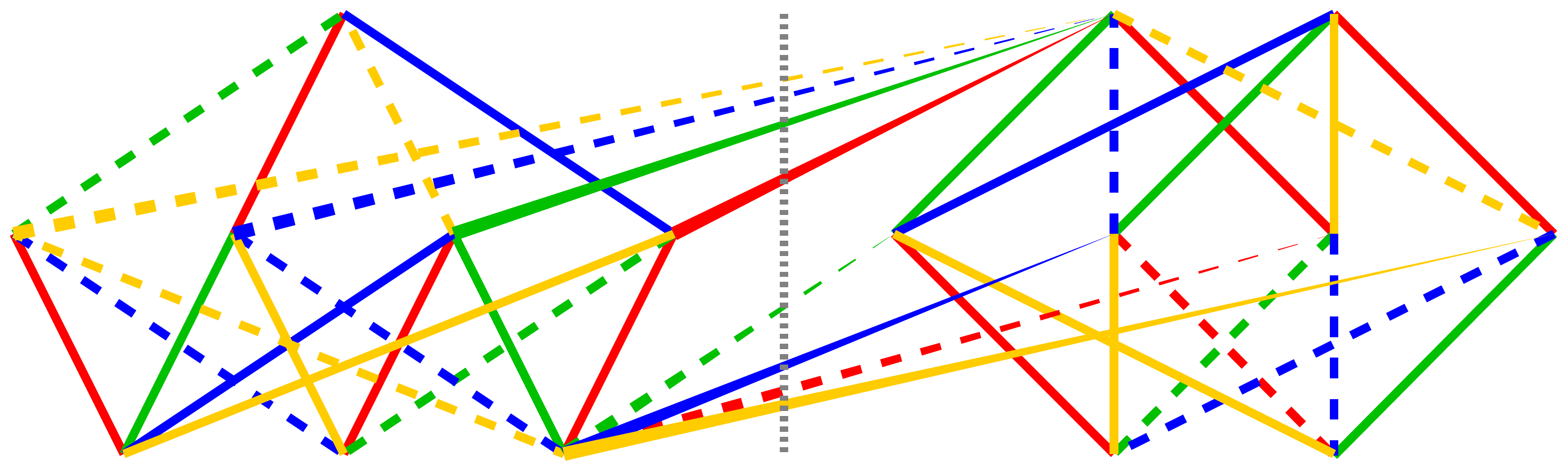}}
    \put(11,1){\cB{$\f_1$}}
    \put(31,1){\cB{$\f_2$}}
    \put(50.5,1){\cB{$\f_3$}}
    \put(99.5,1){\cB{$\f_4$}}
    \put(119,1){\cB{$\f_5$}}
    \put(2,20){\bB{$\j_1$}}
    \put(21,20){\bB{$\j_2$}}
    \put(40,20){\bB{$\j_3$}}
    \put(59,20){\bB{$\j_4$}}
    \put(80,20){\bB{$\j_5$}}
    \put(99,20){\bB{$\j_6$}}
    \put(118,20){\bB{$\j_7$}}
    \put(137,20){\bB{$\j_8$}}
    \put(31,40){\cB{$F_1$}}
    \put(99.5,40){\cB{$F_3$}}
    \put(119,40){\cB{$F_2$}}
 \end{picture}
\caption{A graphical depiction of the $N=4$ worldline supermultiplet\eq{e:GKY4}.}
 \label{f:A4B}
\end{figure}
This illustration makes it clear that the $4^\text{th}$ supersymmetry transformations has been implemented here so that the left- and the right-hand side halves in Figure~\ref{f:A4B} have the same, chiral {\em\/chromotopology\/}\ft{Chromotopology is the topology of the graph, taken together with the coloring of the nodes (boson/fermion), the coloring of the edges and their solidity/dashedness.}\cite{r6-3,r6-3.1}, and are then connected by the one-way $Q$-transformations depicted by the left-to-right, upward tapering edges---the only ones that connect the two halves. To verify this, one easily traces any 4-colored quadrangle that lies entirely within the left-hand half and another that is entirely within the right-hand half, multiplying factors of $+1$ for solid edges and $-1$ for dashed edges\cite{rH-TSS}. For both halves in Figure~\ref{f:A4B}, this product equals
\begin{equation}
  (-1) \times (-1)^{F(\text{start})}\,\ve(\C3{r},\C1{g},\C6{b},\C8{y}),
 \label{e:CP(N4)}
\end{equation}
where $(-1)^{F(\text{start})}=+1$ if we start from a boson and $-1$ if we start from a fermion, and $\ve(\C3{r},\C1{g},\C6{b},\C8{y})$ is the sign of the permutation of the colors of the followed edges as compared to the default \C3{red}-\C1{green}-\C6{blue}-\C8{yellow} order; see Ref.\cite{rH-TSS} for details and proof. In fact, the two halves differ one from another only in the mass-dimension of one of the bosonic nodes: upon raising the $\f_2$ node to the top level (mapping $\f_2\mapsto F_4\Defl\dt\f_2$) followed by the corresponding changes in the transformation rules\eq{e:GKY4} and then changing the sign $F_1\mapsto{-}F_1$, the left-hand side of the graph in Figure~\ref{f:A4B} becomes isomorphic to the upside-down image of the right-hand half.

Direct calculation with\eq{e:GKY4}, detailed in the first half of appendix~\ref{a:SuSyX}, shows that all of the kinetic and super-Zeeman Lagrangian terms in Tables~\ref{t:KE} and~\ref{t:SZ} are $\C8{Q_4}$-supersymmetric\eq{e:GKY4}. Therefore, at least the ``quadratic'' Lagrangians\eq{e:KE}+(\ref{e:SZ}) for the supermultiplet\eq{e:GKY} are in fact automatically $N\,{=}\,4$-supersymmetric.

\paragraph{The Chiral-Twisted-Chiral Extension:}
In distinction from the chiral-chiral extension\eq{e:GKY4}, the supersymmetry transformation rules\eq{e:GKY} may also be extended in a twisted fashion:
\begin{equation}
  \begin{array}{@{} c|c@{~~}c@{~~}c@{~~}c @{}}
 & \C3{Q_1} & \C1{Q_2} & \C6{Q_3} & \C8{\Tw{Q}_4} \\ 
    \hline\rule{0pt}{2.1ex}
\f_1
 & \j_1 & \j_2 & \j_3 & \j_4 \\[0pt]
\f_2
 & \j_3 & -\j_4 & -\j_1 & \j_2 \\[0pt]
\f_3
 & \j_4\C5{{-}\j_7} & \j_3\C5{{-}\j_5}
   & -\j_2\C5{{+}\j_6} & -\j_1\color{grey3}\fbox{$\C5{{-}\j_8}$} \\[-2pt]
\f_4
 & \j_5 & -\j_7 & -\j_8 & \cB{$-\j_6$} \\[0pt]
\f_5
 & -\j_6 & \j_8 & -\j_7 & \cB{$-\j_5$} \\*[3pt]
    \cline{2-5}\rule{0pt}{2.5ex}
F_1
 & \dt\j_2 & -\dt\j_1 & \dt\j_4 & -\dt\j_3 \\[0pt]
F_2
 & \dt\j_8 & \dt\j_6 & \dt\j_5 & \cB{$-\dt\j_7$} \\[0pt]
F_3
 & \dt\j_7 & \dt\j_5 & -\dt\j_6 & \cB{$\dt\j_8$} \\
    \hline
  \end{array}
 \quad~~
  \begin{array}{@{} c|c@{~}c@{~}c@{~}c @{}}
 & \C3{Q_1} & \C1{Q_2} & \C6{Q_3} & \C8{\Tw{Q}_4} \\ 
    \hline\rule{0pt}{2.9ex}
\j_1
 & i\dt\f_1 & -iF_1 & -i\dt\f_2 & -i\dt\f_3\C5{{-}iF_3} \\[0pt]
\j_2
 & iF_1 & i\dt\f_1 & -i\dt\f_3\C5{{-}iF_3} & i\dt\f_2 \\[0pt]
\j_3
 & i\dt\f_2 & i\dt\f_3\C5{{+}iF_3} & i\dt\f_1 & -iF_1 \\[0pt]
\j_4
 & i\dt\f_3\C5{{+}iF_3} & -i\dt\f_2 & iF_1  & i\dt\f_1\\[0pt]
\j_5
 & i\dt\f_4 & iF_3 & iF_2 & \cB{$-i\dt\f_5$} \\[0pt]
\j_6
 & -i\dt\f_5 & iF_2 & -iF_3 & \cB{$-i\dt\f_4$} \\[0pt]
\j_7
 & iF_3 & -i\dt\f_4 & -i\dt\f_5 & \cB{$-iF_2$} \\[0pt]
\j_8
 & iF_2 & i\dt\f_5 & -i\dt\f_4 & \cB{$iF_3$} \\[0pt]
    \hline
  \end{array}
  \label{e:GKY4-}
\end{equation}
which may be depicted in the manner of Figure~\ref{f:A4B}. The entries differing (only in sign) from the corresponding ones in\eq{e:GKY4} have been boxed.
\begin{figure}[ht]
\centering
 \begin{picture}(140,45)
   \put(0,0){\includegraphics[width=140mm]{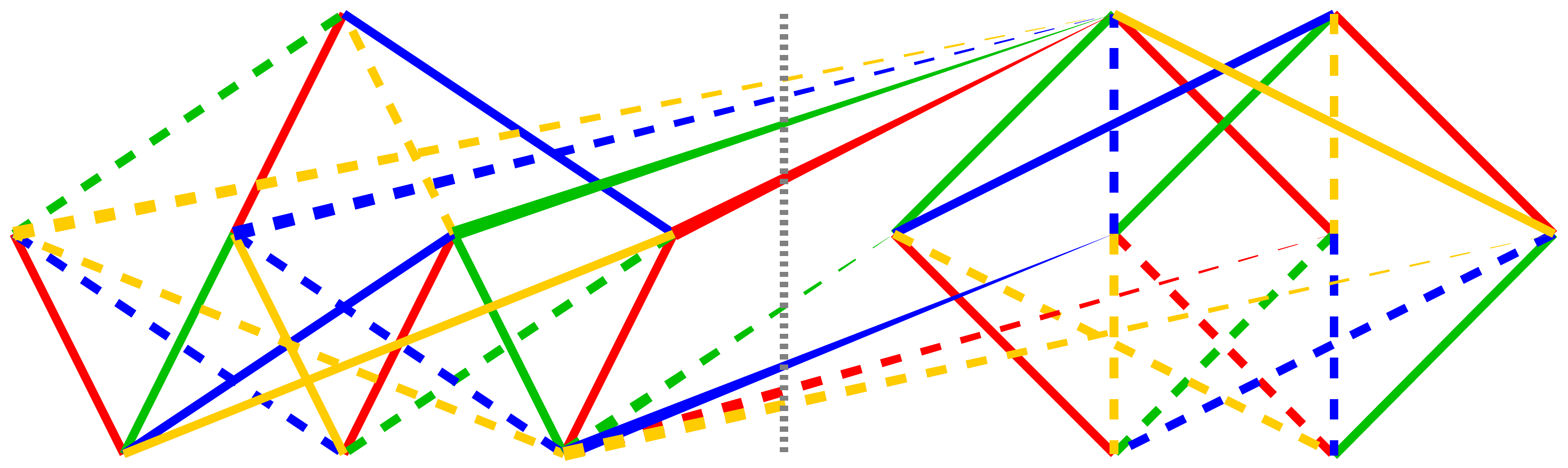}}
    \put(11,1){\cB{$\f_1$}}
    \put(31,1){\cB{$\f_2$}}
    \put(50.5,1){\cB{$\f_3$}}
    \put(99.5,1){\cB{$\f_4$}}
    \put(119,1){\cB{$\f_5$}}
    \put(2,20){\bB{$\j_1$}}
    \put(21,20){\bB{$\j_2$}}
    \put(40,20){\bB{$\j_3$}}
    \put(59,20){\bB{$\j_4$}}
    \put(80,20){\bB{$\j_5$}}
    \put(99,20){\bB{$\j_6$}}
    \put(118,20){\bB{$\j_7$}}
    \put(137,20){\bB{$\j_8$}}
    \put(31,40){\cB{$F_1$}}
    \put(99.5,40){\cB{$F_3$}}
    \put(119,40){\cB{$F_2$}}
 \end{picture}
\caption{A graphical depiction of the $N=4$ worldline supermultiplet\eq{e:GKY4-}.}
 \label{f:A4B-}
\end{figure}
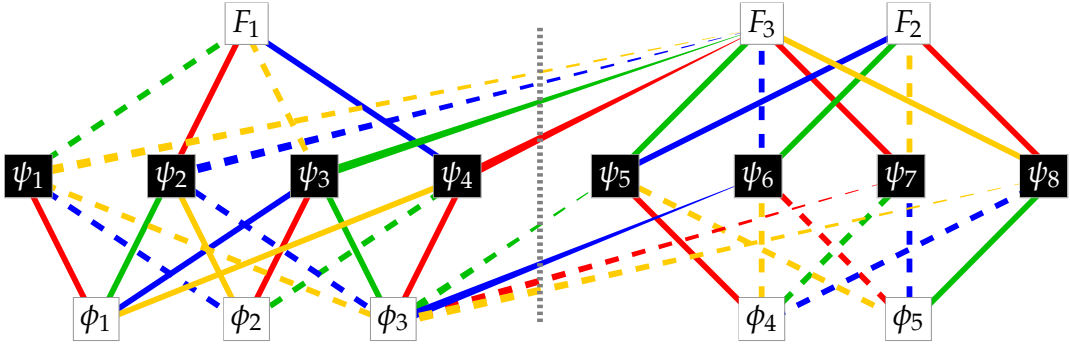
This illustration makes it clear that the $4^\text{th}$ supersymmetry transformations can also be implemented so that the left- and the right-hand side halves in Figure~\ref{f:A4B} differ one from another, besides the mass-dimension of one of the bosonic node, also in chromotopology. We again verify this by tracing 4-colored quadrangles following the definition in Ref.\cite{rH-TSS}, and obtain for the two halves in Figure~\ref{f:A4B-}:
\begin{equation}
  (\mp1) \times (-1)^{F(\text{start})}\,\ve(\C3{r},\C1{g},\C6{b},\C8{y})\quad
  \text{for the}\,\Big\{\!\!\begin{tabular}{l} left-hand\\[-2pt]
                                               right-hand\end{tabular}\!\!\Big\}\,\text{half}.
 \label{e:CP(N4)t}
\end{equation}
 Whereas the chromotopology of the left-hand half is chiral, that of the right-hand half is twisted-chiral\cite{r6-3,r6-3.1}, and the two halves are connected by the one-way $Q$-transformations depicted by the left-to-right, upward tapering edges---the only ones that connect the two halves.

Direct calculation with\eq{e:GKY4-}, detailed in the second half of Appendix~\ref{a:SuSyX}, shows that the first two of the kinetic Lagrangian terms in Table~\ref{t:KE} and the first four super-Zeeman Lagrangian in Table~\ref{t:SZ} are also $\C8{\Tw{Q}_4}$-supersymmetric. However, the last four kinetic Lagrangian terms in Table~\ref{t:KE} and the last three super-Zeeman Lagrangian terms in Table~\ref{t:SZ} are not $\C8{\Tw{Q}_4}$-supersymmetric, and so present an obstruction to the twisted $\C8{\Tw{Q}_4}$-supersymmetry.

That is, the 6-parameter family of quadratic Lagrangians
\begin{equation}
  \sL^{\sss\text{KE}}_{\sss(A_1,A_2,0,0,0,0)}
  +\sL^{\sss\text{SZ}}_{\sss(B_1,B_2,B_3,B_4,0,0,0)}
\end{equation}
are both $\C8{Q_4}$- and $\C8{\Tw{Q}_4}$-supersymmetric. In turn, adding any combination of the $A_3,{\cdots},A_6$- and $B_5,B_6,B_7$-dependent terms obstructs the ``twisted'' $\C8{\Tw{Q}_4}$-supersymmetry\eq{e:GKY4}, but not the ``chiral-chiral'' $\C8{Q_4}$-supersymmetry\eq{e:GKY4-}, and the general bilinear Lagrangian\eq{e:KE}+(\ref{e:SZ}) is indeed $N\,{=}\,4$ supersymmetric.

Finally, it is interesting to note that the {\em\/standard\/} kinetic term Lagrangians\eq{e:KE=L} and\eq{e:KE=R} may in fact be written as
\begin{equation}
 \begin{array}{rclcrcl}
 \sL_{\sss(1,0,0,0,0,0)}^{\sss\text{KE}}
  &=&\C8{Q_4}\C6{Q_3}\C1{Q_2}\C3{Q_1}\,\fc12\f_1^{~2}
  \qquad&\qquad\text{and}&\qquad
 \sL_{\sss(0,1,0,0,0,0)}^{\sss\text{KE}}
  &=&\C8{Q_4}\C6{Q_3}\C1{Q_2}\C3{Q_1}\,\fc12\f_4^{~2}\\
  &=&\C8{\Tw{Q}_4}\C6{Q_3}\C1{Q_2}\C3{Q_1}\,\fc12\f_1^{~2}
  \qquad&&\qquad
  &=&-\C8{\Tw{Q}_4}\C6{Q_3}\C1{Q_2}\C3{Q_1}\,\fc12\f_4^{~2}\\
 \end{array}
\end{equation}
making them manifestly $N\,{=}\,4$-supersymmetric. The remaining terms in Tables~\ref{t:KE} and~\ref{t:SZ} however do not have such simple representations.

\section{Closing Comments}
\label{s:Coda}
In Ref.\cite{r6-1}, we have proven that a semi-infinite iterative sequence of off-shell supermultiplet quotients produces ever larger indecomposable supermultiplets of worldline $N\,{\geqslant}\,3$ supersymmetry. Here, we constructed a 13-parameter family of Lagrangians\eq{e:KE}+(\ref{e:SZ}) bilinear in the component fields of the smallest one of these supermultiplets\eq{e:GKY}, and show that this family contains a continuum of unitary models. Although\eq{e:GKY} is a supermultiplet of $N\,{=}\,3$ supersymmetry by construction, its most general bilinear Lagrangian\eq{e:KE}+(\ref{e:SZ}) turns out to automatically admit a fourth supersymmetry. It remains an open question whether the more general Lagrangians\eq{e:EG} with an unlimited number of higher order interaction terms also admit a fourth supersymmetry.

The 13-parameter family of bilinear Lagrangians\eq{e:KE}+(\ref{e:SZ}) involve a 7-parameter subset of terms that represent linear couplings of external magnetic fluxes to angular momenta corresponding to rotations in planes within the $(\f_1,\cdots,\f_5)$-space. Of these terms, each one of the first five terms in Table~\ref{t:SZ} results in a simple harmonic, helicoidal motion such as\eq{e:Z3}. However, each one of the last two terms in Table~\ref{t:SZ} results generically in a chaotic, aperiodic motion such as\eq{e:f1f3}+(\ref{e:f4f2f5})---the normal modes in this regime have incommensurate frequencies differing by powers of the Golden Ratio\eq{e:2freqs}. The 13-parameter family then clearly contains at least these two phases with radically different oscillatory regimes.

The present analysis focuses on the classical behavior, but in doing so defines the Lagran\-gi\-ans\eq{e:KE}+(\ref{e:SZ})+(\ref{e:EG})+(\ref{e:ZEg}). Owing to the fact that the supermultiplet\eq{e:GKY} is fully off-shell\cite{rTHGK12}, supersymmetry closes on the each component field of the supermultiplet $\F\Defl(\f_i|\j_\hk|F_m)$ unconditionally and it is straightforward to construct the partition functional
\begin{equation}
  Z[\ha\F;\vec{A},\vec{B},\vec{C}\,]
   \Defl \pmb\int\text{\bsf D[$\F$]}~
          \exp\Big\{ \int\rd\t~\Big(\sL_{\sss\vec{A}}^{\sss\text{KE}}
                                  {+}\sL_{\sss\vec{B}}^{\sss\text{SZ}}
                                   {+}\sL_{\sss\vec{C}}^{\sss\text{EG}}
                                     + \ha\F{\cdot}\F\Big)\Big\},
\end{equation}
where $\ha\F\Defl(\ha\f^i|\ha\j^\hk|\ha{F}^m)$ is a corresponding supermultiplet of probing sources. The methods used in Sections~\ref{s:SSL} and~\ref{s:MSI} can equally well provide interactions of the supermultiplet\eq{e:GKY} with other off-shell supermultiplets\cite{rKT07,r6-3.1}; this we defer to a subsequent effort. For illustration, suffice it here to note that the Hamiltonian for the $(\f_1,\f_3)$ bosonic system\eq{e:L6f} may be written, after integration by parts as in $\fc12B_6\f_1\dt\f_3\simeq\fc14B_6(\f_1\dt\f_3{-}\dt\f_1\f_3)$:
\begin{equation}
  \begin{aligned}
    \mathscr{H}_{\sss13}
   &= \inv2\big(p_1{+}\inv4B_6\f_3\big)^2 + \big(p_3{-}\inv4B_6\f_1\big)^2
      + \inv4B_6^{~2}\f_1^{~2} + \inv2B_6^{~2}\f_3^{~2},\\
   &= \inv2 p_{\sss13}^{~2}
       +\inv2\big(\frc{\ell_{\sss13}}{\vf_{13}}
            {-}\frc{B_6}{2\sqrt2}\vf_{\sss13}\big)^2
         +\inv8B_6^{~2}\big(5{-}3\cos(2\vq_{\sss13})\big)\vf_{\sss13}^{~2},
  \end{aligned}
 \label{e:KESZEG}
\end{equation}
where $\f_1=\vf_{\sss13}\cos(\vq_{\sss13})$ and $\f_3=\sqrt2\vf_{\sss13}\sin(\vq_{\sss13})$, and $p_{\sss13}$ and $\ell_{\sss13}$ are the radial  and angular momenta canonically conjugate to $\vf_{\sss13}$ and $\vq_{\sss13}$, respectively. This clearly indicates the nontrivial dynamics within even this very simple subsystem of\eq{e:KESZ}+(\ref{e:EG}), and motivates the more elaborate path integral approach in\eq{e:KESZEG}.

\bigskip
\paragraph{\bfseries Acknowledgments:}
 TH is grateful to
 the Department of Physics, University of Central Florida, Orlando FL,
 and
 the Physics Department of the Faculty of Natural Sciences of the University of Novi Sad, Serbia, 
 for the recurring hospitality and resources.
 GK is grateful to
 the Department of Physics and Astronomy, Howard University, Washington DC, 
 for the hospitality and resources.

\appendix
\section{Technical Details}
\label{s:Bits}
The following technical details have been deferred from the main narrative for clarity.

\subsection{Kinetic Terms}
\label{a:KE}
The computation of the kinetic Lagrangian terms\eq{e:L0} proceeds by straightforward use of the transformations specified in Table~\ref{e:GKY},  For example:
\begin{equation}
 \C3{Q_1}\C1{Q_2}\C6{Q_3}(\f_2\j_2)
  \simeq-\big[\, \dt\f_1^{~2} + \dt\f_2^{~2} +(\dt\f_3+F_3)^2 + F_1^{~2}
               + i\j_1\dt\j_1 + i\j_2\dt\j_2 + i\j_3\dt\j_3 + i\j_4\dt\j_4\,\big],
 \label{e:KE11}
\end{equation}
which appears in the first row of Table~\ref{t:KE}. The relation ``$\simeq$'' denotes equality up to total time-derivatives, but we have {\em\/not\/} made use of related (anti)symmetrizations such as
\begin{equation}
\begin{aligned}
   i\j_2\dt\j_2
   &= \frc{i}2(\j_2\,\dt\j_2-\dt\j_2\,\j_2) + \frc{i}2(\j_2\,\dt\j_2+\dt\j_2\,\j_2)
    = \frc{i}2(\j_2\,\dt\j_2-\dt\j_2\,\j_2) + i\vdt(\j_2\,\j_2~=~0)\\
   &= \frc{i}2(\j_2\,\dt\j_2-\dt\j_2\,\j_2),\quad\etc
\end{aligned}
\end{equation}
The terms\eq{e:KE11} clearly form one specific subset of the terms in the Lagrangian\eq{e:KE}.

The forty manifestly supersymmetric terms $\C6{Q_3}\C1{Q_2}\C3{Q_1}(\f_i\j_\hk)$ are listed in Table~\ref{t:KE0}, expanded and reduced modulo total time-derivatives.
\begin{table}[htpb]
\small
$$
  \begin{array}{@{} r@{:~}l @{}}
\bs{\f_i\j_\hk}&\text{\bsf The Kinetic Lagrangian Terms \boldmath
    $Q^3(\f_i\j_\hk)\Defl\C6{Q_3}\C1{Q_2}\C3{Q_1}(\f_i\j_\hk)$}
 \\[1pt]
    \toprule
\f_1 \j_1&
 \C5{\vdt\big[ \f_1(F_3{+}\dt\f_3) - i\j_2 \j_3 \big]}
 \\
\f_1 \j_2&
\C5{-\vdt\big[ \f_1 \dt\f_2 - i\j_1\j_3 \big]}
 \\
\f_1 \j_3& 
\C5{\vdt\big[ F_1 \f_1 - i\j_1 \j_2 \big]}
 \\
\f_1 \j_4& 
   F_1^{~2}
 + (F_3{+}\dt\f_3)^2
 + \dt\f_1 \dt\f_1  
 + \dt\f_2 \dt\f_2 
 + i\j_1 \dt\j_1 
 + i\j_2 \dt\j_2 
 + i\j_3 \dt\j_3 
 + i\j_4 \dt\j_4
\,\C5{-\vdt[\f_1\dt\f_1
 \big]}
 \\
\f_1 \j_5& 
   F_1 F_2
 - F_3 (\dt\f_2{-}\dt\f_4)
 + \dt\f_1 \dt\f_5  
 + \dt\f_3 \dt\f_4 
 - i\j_1 \dt\j_6 
 + i\j_2 \dt\j_8 
 - i\j_3 \dt\j_7 
 + i\j_4 \dt\j_5 
\,\C5{-\vdt[
   \f_1 \dt\f_5
 + i\j_4\j_5
 \big]}
 \\
\f_1 \j_6&
 - F_2 \dt\f_2
 - F_3 (F_1{+}\dt\f_5)
 + \dt\f_1 \dt\f_4 
 - \dt\f_3 \dt\f_5 
 + i\j_1 \dt\j_5 
 - i\j_2 \dt\j_7 
 - i\j_3 \dt\j_8 
 + i\j_4 \dt\j_6 
\,\C5{-\vdt[
   \f_1\dt\f_4 
 + i\j_4\j_6
 \big]}
 \\
\f_1 \j_7&
 - F_1 \dt\f_5 
 + F_2 \dt\f_1
 + F_3(F_3{+}\dt\f_3)
 + \dt\f_2 \dt\f_4 
 + i\j_1 \dt\j_8 
 + i\j_2 \dt\j_6 
 + i\j_3 \dt\j_5 
 + i\j_4 \dt\j_7 
\,\C5{-\vdt[
   F_2\f_1
 + i\j_4\j_7
 \big]}
\\
\f_1 \j_8&
 - F_1 \dt\f_4 
 + F_2(F_3{+}\dt\f_3)
 - F_3 \dt\f_1
 - \dt\f_2 \dt\f_5 
 - i\j_1 \dt\j_7 
 - i\j_2 \dt\j_5 
 + i\j_3 \dt\j_6 
 + i\j_4 \dt\j_8
\,\C5{+\vdt[
   F_3\f_1
 - i\j_4\j_8
 \big]}
 \\ \hline\noalign{\vglue1pt}
\f_2\j_1 & 
\,\C5{\vdt\big[
   F_3\f_2
 + \f_2\dt\f_3  
 - i\j_3\j_4 
 \big]}
 \\
\f_2\j_2 &
~=Q^3(\f_1\j_4)
\,\C5{+\vdt\big[
   \f_1 \dt\f_1
 - \f_2 \dt\f_2
 \big]}
 \\
\f_2\j_3 &
 \,\C5{\vdt\big[
   F_1\f_2 
 + i\j_1\j_4 
 \big]}
 \\
\f_2\j_4 &
\,\C5{-\vdt\big[
   \dt\f_1 \f_2
 + i\j_1\j_3 
 \big]}
 \\
\f_2\j_5 &
~=-Q^3(\f_1\j_8)
\,\C5{+\vdt\big[
   F_3\f_1
 - \f_2 \dt\f_5  
 - i\j_2 \j_5 
 - i\j_4\j_8
 \big]}
 \\ 
\f_2\j_6 &
~=Q^3(\f_1\j_7)
\,\C5{+\vdt\big[
   F_2\f_1
 - \f_2 \dt\f_4  
 - i\j_2 \j_6 
 + i\j_4\j_7
 \big]}
 \\
\f_2\j_7 &
~=-Q^3(\f_1\j_6)
\,\C5{-\vdt\big[
   \f_1\dt\f_4 
 + F_2 \f_2
 + i\j_2 \j_7 
 + i\j_4\j_6
 \big]}
 \\
\f_2\j_8 &
~=Q^3(\f_1\j_5)
\,\C5{+\vdt\big[
   F_3 \f_2
 + \f_1 \dt\f_5
 - i\j_2 \j_8 
 + i\j_4\j_5
 \big]}
  \\ \hline\noalign{\vglue1pt}
\f_3 \j_1&
 - F_1(F_1{+}\dt\f_5)
 + F_2 \dt\f_1 
 - (F_3{+}\dt\f_3)\dt\f_3  
 - \dt\f_1 \dt\f_1 
 - \dt\f_2(\dt\f_2{-}\dt\f_4)
\,\C5{+\vdt\big[
   (F_3{+}\dt\f_3)\f_3
 - i\j_2 \j_6
 - i\j_3 \j_5
 - i\j_4 \j_7
  \big]}\\*[-1pt]
\omit&~~
 - i\j_1(\dt\j_1{-}\dt\j_8)
 - i\j_2(\dt\j_2{-}\dt\j_6)
 - i\j_3(\dt\j_3{-}\dt\j_5)
 - i\j_4(\dt\j_4{-}\dt\j_7)
\\
\f_3 \j_2&
~=Q^3(\f_1\j_5)
\,\C5{+\vdt\big[
   \f_1 \dt\f_5
 - \dt\f_2 \f_3
 + i\j_1 \j_6
 + i\j_3 \j_7
 \big]}
\\
\f_3 \j_3&
~=-Q^3(\f_1\j_6)
\,\C5{+\vdt\big[
   F_1\f_3 
 - \f_1\dt\f_4 
 + i\j_1\j_5
 + i\j_2(\j_4{-}\j_7)
 \big]}
\\
\f_3 \j_4&
~=Q^3(\f_1\j_8)
\,\C5{-\vdt\big[
   (F_3{-}\dt\f_3)\f_1
 + \vdt(\f_1\f_3)
 - i\j_1 \j_7
 + i\j_2(\j_3{-}\j_5)
 + i\j_3 \j_6 
 - i\j_4\j_8
 \big]}
\\
\f_3 \j_5&
~=-Q^3(\f_1\j_6)
\,\C5{-\vdt\big[
   \f_1\dt\f_4 
 + \f_3 \dt\f_5  
 - i\j_1 \dt\j_5
 + i(\j_4{+}\j_7)\j_6
 \big]}
\\
\f_3 \j_6&
~=Q^3(\f_1\j_5)
\,\C5{+\vdt\big[
   \f_1 \dt\f_5
 - \f_3 \dt\f_4  
 + i\j_1 \j_6
 + i(\j_4{-}\j_7)\j_5
 \big]}
\\
\f_3 \j_7&
~=Q^3(\f_1\j_8)
\,\C5{-\vdt\big[
   F_2 \f_3
 + F_3\f_1
 - i\j_1 \j_7
 - i\j_4\j_8
 + i\j_5 \j_6 
 \big]}
\\
\f_3 \j_8&
   F_1 \dt\f_5 
 + F_2(F_2{-}\dt\f_1)
 - F_3 \dt\f_3
 - \dt\f_2 \dt\f_4 
 + \dt\f_4 \dt\f_4 
 + \dt\f_5 \dt\f_5
\,\C5{+\vdt\big[
   F_3 \f_3
 + i\j_1 \j_8
  \big]}
\\*[-1pt]
\omit&~~
 - i(\j_1{-}\j_8)\dt\j_8
 - i(\j_2{-}\j_6)\dt\j_6 
 - i(\j_3{-}\j_5)\dt\j_5 
 - i(\j_4{-}\j_7)\dt\j_7 
\\[1pt] \hline\noalign{\vglue1pt}
\f_4 \j_1&
~=-Q^3(\f_1\j_5)
\,\C5{+\vdt\big[
   (F_3{+}\dt\f_3)\f_4
 - \f_1 \dt\f_5
 + i\j_2 \j_8
 - i\j_3 \j_7
 \big]}
\\
\f_4 \j_2&
~=Q^3(\f_1\j_7)
\,\C5{+\vdt\big[
   F_2\f_1
 - \f_4 \dt\f_2
 - i\dt\j_1 \j_8
 - i\j_3 \j_5
 \big]}
 \\
\f_4 \j_3&
~=Q^3(\f_1\j_8)
\,\C5{+\vdt\big[
   F_1 \f_4
 - F_3 \f_1
 + i\j_1 \j_7
 + i\j_2 \j_5
 \big]}
 \\
\f_4 \j_4&
~=Q^3(\f_1\j_6)
\,\C5{+\vdt\big[
   \f_1 \dt\f_4 
 - \dt\f_1 \f_4 
 - i\j_1 \j_5
 + i\j_2 \j_7
 + i\j_3 \j_8
 + i\j_4 \j_6
 \big]}
 \\
\f_4 \j_5&
\,\C5{-\vdt\big[
   \f_4 \dt\f_5  
 + i\j_7 \j_8 
  \big]}
\\
\f_4 \j_6&
   F_2^{~2}
 + F_3^{~2}
 + \dt\f_4 \dt\f_4  
 + \dt\f_5 \dt\f_5 
 + i\j_5 \dt\j_5 
 + i\j_6 \dt\j_6 
 + i\j_7 \dt\j_7 
 + i\j_8 \dt\j_8 
\,\C5{-\vdt\big[
   \f_4 \dt\f_4  
  \big]}
\\
\f_4 \j_7&
\,\C5{-\vdt\big[
   F_2 \f_4 
 - i\j_5\j_8 
  \big]}
\\
\f_4 \j_8&
\,\C5{\vdt\big[
   F_3 \f_4
 - i\j_5 \j_7 
  \big]}
\\ \hline\noalign{\vglue1pt}
\f_5 \j_1&
~=Q^3(\f_1\j_6)
\,\C5{+\vdt\big[
   F_3 \f_5
 + \f_1 \dt\f_4 
 + \dt\f_3 \f_5
 + i\j_2 \j_7
 + i\j_3 \j_8
 \big]}
\\
\f_5 \j_2&
~=-Q^3(\f_1\j_8)
\,\C5{+\vdt\big[
   F_3 \f_1
 - \dt\f_2 \f_5
 - i\j_1 \j_7
 + i\j_3 \j_6
 \big]}
\\
\f_5 \j_3&
~=Q^3(\f_1\j_7)
\,\C5{+\vdt\big[
   F_1 \f_5
 + F_2\f_1
 - i\j_1 \j_8
 - i\j_2 \j_6
 \big]}
 \\
\f_5 \j_4&
~=Q^3(\f_1\j_5)
\,\C5{+\vdt\big[
 + \f_1 \dt\f_5
 - \dt\f_1 \f_5
 + i\j_1 \j_6
 - i\j_2 \j_8
 + i\j_3 \j_7
 + i\j_4\j_5
 \big]}
 \\
\f_5 \j_5&
~=Q^3(\f_4\j_6)
\,\C5{+\vdt\big[
   \f_4 \dt\f_4  
 + \f_5 \dt\f_5  
 \big]}
\\
\f_5 \j_6&
\,\C5{-\vdt\big[
   \dt\f_4 \f_5
 - i\j_7 \j_8 
  \big]}
\\
\f_5 \j_7&
\,\C5{-\vdt\big[
   F_2 \f_5
 + i\j_6 \j_8
  \big]}
 \\
\f_5 \j_8&
\,\C5{\vdt\big[
   F_3\f_5 
 + i\j_6 \j_7
  \big]}
\\
    \bottomrule
  \end{array}
$$\vspace{-7mm}
  \caption{The manifestly supersymmetric expressions $\C6{Q_3}\C1{Q_2}\C3{Q_1}(\f_i\j_\hk)$, with indicated total time derivatives}
  \label{t:KE0}
\end{table}
Several of the entries turn out to {\em\/be\/} total time-derivatives, and so produce a vanishing entry. To illustrate this, consider the case when $\j_\hk$ is in fact a supersymmetric partner of $\f_i$:
\begin{alignat}9
 \C6{Q_3}\C1{Q_2}\C3{Q_1}(\f_2\j_1)
 &=\C1{Q_2}\C3{Q_1}\C6{Q_3}\big[\f_2\,(-\C6{Q_3}\f_2)\big]
  =-\C1{Q_2}\C3{Q_1}\C6{Q_3}\big[\fc12\C6{Q_3}(\f_2^{\,2})\big]
  =-\frc{i}2\vdt[\C1{Q_2}\C3{Q_1}\,\f_2^{\,2}]~\simeq~0.
 \label{e:KEtriv}
\intertext{Furthermore, Table~\ref{t:KE0} contains many redundancies, owing to identities of the following type:}
 \C6{Q_3}\C1{Q_2}\C3{Q_1}(\f_2\j_6)
 &= \C6{Q_3}\C1{Q_2}\C3{Q_1}\big[\f_2\,(-\C3{Q_1}\f_5)\big]
  = \C6{Q_3}\C1{Q_2}\C3{Q_1}\big[-\C3{Q_1}(\f_2\,\f_5)+(\C3{Q_1}\f_2)\,\f_5\big],\nn\\
 &=-\C3{Q_1}^2\C6{Q_3}\C1{Q_2}\big[\f_2\,\f_5\big]
    +\C6{Q_3}\C1{Q_2}\C3{Q_1}\big[(\C3{Q_1}\f_2)\,\f_5\big],\nn\\
 &=-i\vdt\big[\C6{Q_3}\C1{Q_2}(\f_2\,\f_5)\big]
    +\C6{Q_3}\C1{Q_2}\C3{Q_1}\big[\j_3\,\f_5\big]
 ~\simeq~\C6{Q_3}\C1{Q_2}\C3{Q_1}\big[\f_5\j_3\big].
 \label{e:KErels}
\end{alignat}
These have been employed to reduce Table~\ref{t:KE0} in the appendix to Table~\ref{t:KE}, in the narrative.

There also exist more complicated relations, such as can be detected on expanding:
\begin{alignat}9
 \C6{Q_3}&\C1{Q_2}\C3{Q_1}\big[\f_2\j_2 - \f_4\j_6\big]\nn\\
 &= \Big[\big( F_1^{~2}
              +\dt\f_2\dt\f_2
              +(F_3{+}\dt\f_3)^2
              +\dt\f_1\dt\f_1 \big)
       +i\big( \j_2\dt\j_2 
              +\j_1\dt\j_1
              +\j_4\dt\j_4
              +\j_3\dt\j_3 \big)\Big] \nn\\
 &\qquad
   -\Big[\big( F_2^{~2}
              +F_3^{~2}
              +\dt\f_4\dt\f_4
              +\dt\f_5\dt\f_5 \big)
       +i\big( \j_6\dt\j_6 
              +\j_8\dt\j_8
              +\j_7\dt\j_7 
              +\j_5\dt\j_5 \big)\Big],
 \label{e:KE=L+R}\\
 \makebox[0pt][r]{and~}
 -\C6{Q_3}&\C1{Q_2}\C3{Q_1}\big[\f_3(\j_1{+}\j_8)\big]\nn\\
 &= \Big[\big( F_1^{~2}
              +\dt\f_2\dt\f_2
              +(F_3{+}\dt\f_3)^2
              +\dt\f_1\dt\f_1 \big)
       +i\big( \j_2\dt\j_2 
              +\j_1\dt\j_1
              +\j_4\dt\j_4
              +\j_3\dt\j_3 \big)\Big] \nn\\
 &\qquad
   -\Big[\big( F_2^{~2}
              +F_3^{~2}
              +\dt\f_4\dt\f_4
              +\dt\f_5\dt\f_5 \big)
       +i\big( \j_6\dt\j_6 
              +\j_8\dt\j_8
              +\j_7\dt\j_7 
              +\j_5\dt\j_5 \big)\Big].
 \label{e:KE=L-R}
\end{alignat}
 The equality of the expansions\eqs{e:KE=L+R}{e:KE=L-R} exhibits the relation
\begin{equation}
  \C6{Q_3}\C1{Q_2}\C3{Q_1}\big[(\f_2\j_2{-}\f_4\j_6)+(\f_3\j_1{+}\f_3\j_8)\big] = 0,
 \label{e:Xr1}
\end{equation}
whereby we can drop for example the term $\C6{Q_3}\C1{Q_2}\C3{Q_1}(\f_3\j_1)$ as being linearly dependent upon the three other three, as indicated\eq{e:Xr1}.
 A similar identity is found by iteratively expanding:
\begin{alignat}9
 \C6{Q_3}&\C1{Q_2}\C3{Q_1}\big[\f_4\j_6 - \f_2\j_6 - \f_3\j_8\big]
 &=0.
\end{alignat}
Indeed, this last identity follows more simply from observing that
\begin{align}
  \C3{Q_1}&\C1{Q_2}\C6{Q_3}\big[\f_4\j_6-\f_2\j_6-\f_3\j_8\big]
   =\C6{Q_3}\C1{Q_2}\C3{Q_1}
     \big[\f_4(-\C3{Q_1}\f_5)-\f_2(-\C3{Q_1}\f_5)-\f_3(\C1{Q_2}\f_5)\big],\nn\\
  &\simeq\C6{Q_3}\C1{Q_2}\C3{Q_1}
     \big[(\C3{Q_1}\f_4)\f_5-(\C3{Q_1}\f_2)\f_5+(\C1{Q_2}\f_3)\f_5\big],\nn\\
  &\quad=\C6{Q_3}\C1{Q_2}\C3{Q_1}
     \big[\big(\j_5-\j_3+(\j_3-\j_5)\big)\f_5\big]
   =\C6{Q_3}\C1{Q_2}\C3{Q_1}
     \big[(\j_5-\j_3+\j_3-\j_5)\f_5\big]=0.
 \label{e:Xr2}
\end{align}
As before, ``$\simeq$'' denotes equality up to omitted total time-derivatives.
 This then allows dropping also $\C6{Q_3}\C1{Q_2}\C3{Q_1}(\f_3\j_8)$, it being a linear combination of $\C6{Q_3}\C1{Q_2}\C3{Q_1}(\f_2\j_6)$ and $\C6{Q_3}\C1{Q_2}\C3{Q_1}(\f_4\j_6)$.

These identities leave the final collection of kinetic terms as listed in Table~\ref{t:KE}.

\subsection{Super-Zeeman Terms}
\label{a:ZE}
The approach followed in Section~\ref{s:ZE} requires the 15 expressions $\C6{Q_3}\C1{Q_2}\C3{Q_1}(\f_i\f_k)$, presented in Table~\ref{t:ZE0}.
\begin{table}[htpb]
$$
  \begin{array}{@{} r@{:~}l @{}}
\bs{\f_i\f_k}&\text{\bsf\boldmath Expression
    $Q^3(\f_i\f_k)\Defl\C6{Q_3}\C1{Q_2}\C3{Q_1}(\f_i\f_k)$}
 \\[1pt]
    \toprule
\frc12\f_1\f_1
&
 -iF_1\j_3 
 -iF_3\j_1 
 +i\dt\f_1\j_4 
 +i\dt\f_2\j_2
 -i\dt\f_3\j_1 
\\
\f_1\f_2
&
\C5{-i\vdt(\f_1\j_2+\f_2\j_4)}
\\
\f_1\f_3
&
 -iF_1\j_6 
 -iF_2\j_1
 -iF_3(\j_4{-}\j_7)
 +i\dt\f_1\j_8 
 -i\dt\f_2\j_5 
 +i\dt\f_3\j_7 
 +i\dt\f_4\j_3
 -i\dt\f_5\j_2
\\
\f_1\f_4
&
  iF_1\j_8 
 -iF_2\j_2
 +iF_3(\j_3{-}\j_5)
 +i\dt\f_1\j_6 
 -i\dt\f_2\j_7 
 -i\dt\f_3\j_5 
 +i\dt\f_4\j_4
 +i\dt\f_5\j_1
\\
\f_1\f_5
&
  iF_1\j_7 
 -iF_2\j_3
 -iF_3(\j_2{-}\j_6)
 +i\dt\f_1\j_5 
 +i\dt\f_2\j_8 
 +i\dt\f_3\j_6 
 -i\dt\f_4\j_1
 +i\dt\f_5\j_4
\\
\frc12\f_2\f_2
&
 -iF_1\j_3 
 -iF_3\j_1 
 +i\dt\f_1\j_4 
 +i\dt\f_2\j_2
 -i\dt\f_3\j_1
~\C5{=\frc12Q^3(\f_1\f_1)}
\\
\f_2\f_3
&
  iF_1\j_7 
 -iF_2\j_3
 -iF_3(\j_2{-}\j_6)
 +i\dt\f_1\j_5 
 +i\dt\f_2\j_8 
 +i\dt\f_3\j_6 
 -i\dt\f_4\j_1
 +i\dt\f_5\j_4
~\C5{=Q^3(\f_1\f_5)}
\\
\f_2\f_4
&
 -iF_1\j_5 
 +iF_2\j_4
 -iF_3(\j_1{+}\j_8)
 +i\dt\f_1\j_7 
 +i\dt\f_2\j_6 
 -i\dt\f_3\j_8 
 +i\dt\f_4\j_2
 +i\dt\f_5\j_3
\\
\f_2\f_5
&
  iF_1\j_6 
 +iF_2\j_1
 +iF_3(\j_4{-}\j_7)
 -i\dt\f_1\j_8 
 +i\dt\f_2\j_5 
 -i\dt\f_3\j_7 
 -i\dt\f_4\j_3
 +i\dt\f_5\j_2
~\C5{=-Q^3(\f_1\f_3)}
\\
\frc12\f_3\f_3
&
 -iF_1(\j_3{-}\j_5)
 -iF_2(\j_4{-}\j_7)
 \\[-1pt]
 \omit&~
 +i\dt\f_1(\j_4{-}\j_7)
 +i\dt\f_2(\j_2{-}\j_6)
 -i\dt\f_3(\j_1{-}\j_8)
 -i\dt\f_4(\j_2{-}\j_6)
 -i\dt\f_5(\j_3{-}\j_5)
\\
\f_3\f_4
&
  iF_1\j_7
 -iF_2\j_3
 -iF_3(\j_2{-}\j_6)
 +i\dt\f_1\j_5 
 +i\dt\f_2\j_8 
 +i\dt\f_3\j_6 
 -i\dt\f_4\j_1
 +i\dt\f_5\j_4
~\C5{=Q^3(\f_1\f_5)}
\\
\f_3\f_5
&
 -iF_1\j_8 
 +iF_2\j_2
 -iF_3(\j_3{-}\j_5)
 -i\dt\f_1\j_6 
 +i\dt\f_2\j_7 
 +i\dt\f_3\j_5 
 -i\dt\f_4\j_4 
 -i\dt\f_5\j_1
~\C5{=-Q^3(\f_1\f_4)}
\\
\frc12\f_4\f_4
&
  iF_2\j_7 
 -iF_3\j_8
 +i\dt\f_4\j_6 
 +i\dt\f_5\j_5
\\
\f_4\f_5
&
~\C5{-i\vdt(\f_4\j_5+\f_5\j_6)}
\\
\frc12\f_5\f_5
&
  iF_2\j_7 
 -iF_3\j_8 
 +i\dt\f_5\j_5 
 +i\dt\f_4\j_6 
~\C5{=\frc12Q^3(\f_4\f_4)}
\\
\bottomrule
  \end{array}
$$\vspace{-9mm}
  \caption{The $\C6{Q_3}\C1{Q_2}\C3{Q_1}$-transformation of the monomials $(\f_i\f_k)$, with indicated total time derivatives}
  \label{t:ZE0}
\end{table}
Of these 15 expressions, two already are total time derivatives and simple row operations were used to find seven linear combinations that are also total time derivatives; a suitable basis for these is shown in\eq{e:pZb}. Applying $Q_IQ_J$ on these nine combinations, we obtain the twenty-seven supersymmetric candidate Lagrangian summands listed in Tables~\ref{t:ZE01}--\ref{t:ZE03}, below. Inspection and a few simple row operations easily show that these reduce to the seven linearly independent super-Zeeman Lagrangian terms listed in Table~\ref{t:SZ}.
\begin{table}[htpb]
\small$$
  \begin{array}{@{} r@{:~}l @{}}
    \toprule
\fc12(\f_1\f_1{-}\f_2\f_2)
&
 -i\j_1\j_2
 -i\j_3\j_4
 +\f_1F_1
 +\f_2F_3
 +\f_2\dt\f_3
\\
\f_1\f_2
&
 +i\j_1\j_4
 +i\j_2\j_3
 -\f_1F_3
 +\f_2F_1
 -\f_1\dt\f_3
\\
\f_1\f_3{+}\f_2\f_5
&
 -i\j_1\j_3
 +i\j_1\j_5
 +i\j_2\j_4
 -i\j_2\j_7
 -i\j_3\j_8
 +i\j_4\j_6
 \\[-1pt]
\omit&~~
 +\f_2F_2
 +\f_3F_1
 -\f_5F_3
 +\f_1\dt\f_2
 -\f_1\dt\f_4
 +\f_3\dt\f_5
\\
\f_1\f_4{+}\f_3\f_5
&
 +i\j_1\j_7
 +i\j_2\j_5
 -i\j_3\j_6
 -i\j_4\j_8
 +i\j_5\j_6
 +i\j_7\j_8
 \\[-1pt]
\omit&~~
 -\f_1F_3
 +\f_3F_2
 +\f_4F_1
 -\f_2\dt\f_5
 +\f_4\dt\f_5
\\
\f_1\f_5{-}\f_2\f_3
&
 -i\j_1\j_8
 -i\j_2\j_6
 -i\j_3\j_5
 -i\j_4\j_7
 +\f_1F_2
 +\f_3F_3
 +\f_5F_1
 +\f_2\dt\f_4
\\
\f_1\f_5{-}\f_3\f_4
&
 -i\j_1\j_8
 -i\j_2\j_6
 -i\j_3\j_5
 -i\j_4\j_7
 +\f_1F_2
 +\f_3F_3
 +\f_5F_1
 +\f_2\dt\f_4
\\
\omit&~\C5{=i\C1{Q_2}\C3{Q_1}(\f_1\f_5{-}\f_2\f_3)}
\\
\fc12(\f_4\f_4{-}\f_5\f_5)
&
 +i\j_5\j_7
 -i\j_6\j_8
 -\f_4F_3
 -\f_5F_2
\\
\f_4\f_5
&
 -i\j_5\j_8
 -i\j_6\j_7
 +\f_4F_2
 -\f_5F_3
\\
\fc12(\f_1\f_1{-}\f_3\f_3{+}\f_5\f_5){-}\f_2\f_4
&
 -i\j_1\j_2
 -i\j_3\j_4
 +\f_1F_1
 +\f_2F_3
 +\f_2\dt\f_3
~-i\j_5\j_7
 +i\j_6\j_8
 +\f_4F_3
 +\f_5F_2
\\
\omit&
~\C5{=i\C1{Q_2}\C3{Q_1}\big[\fc12(\f_1^{~2}{-}\f_2^{~2}){-}\fc12(\f_4^{~2}{-}\f_5^{~2})\big]}
\\
\bottomrule
  \end{array}
$$\vspace{-9mm}
  \caption{Candidate Lagrangian summands $i\C1{Q_2}\C3{Q_1}\,Z(\f_*)$}
  \label{t:ZE01}
\end{table}
\begin{table}[htpb]
\small$$
  \begin{array}{@{} r@{:~}l @{}}
    \toprule
\fc12(\f_1\f_1{-}\f_2\f_2)
&
~\C5{i\vdt(\f_1\f_2)}
\\
\f_1\f_2
&
~\C5{\frc{i}2\vdt(\f_2^{~2}{-}\f_1^{~2})}
\\
\f_1\f_3{+}\f_2\f_5
&
 +i\j_1\j_2
 +i\j_3\j_4
 -\f_1F_1
 -\f_2F_3
 -\f_2\dt\f_3
\\
\f_1\f_4 +\f_3\f_5
&
 +i\j_1\j_8
 +i\j_2\j_6
 +i\j_3\j_5
 +i\j_4\j_7
 -\f_1F_2
 -\f_3F_3
 -\f_5F_1
 -\f_2\dt\f_4
\\
\f_1\f_5{-}\f_2\f_3
&
 +i\j_1\j_4
 +i\j_2\j_3
 -\f_1F_3
 +\f_2F_1
 -\f_1\dt\f_3
\\
\f_1\f_5{-}\f_3\f_4
&
 +i\j_1\j_7
 +i\j_2\j_5
 -i\j_3\j_6
 -i\j_4\j_8
 +i\j_5\j_6
 +i\j_7\j_8
 \\[-1pt]
\omit&~~
 -\f_1F_3
 +\f_3F_2
 +\f_4F_1
 -\f_2\dt\f_5
 +\f_4\dt\f_5
\\
\fc12(\f_4\f_4{-}\f_5\f_5)
&
 +i\j_5\j_8
 +i\j_6\j_7
 -\f_4F_2
 +\f_5F_3
\\
\f_4\f_5
&
 +i\j_5\j_7
 -i\j_6\j_8
 -\f_4F_3
 -\f_5F_2
\\
\fc12(\f_1\f_1{-}\f_3\f_3{+}\f_5\f_5){-}\f_2\f_4
&
 -i\j_1\j_3
 +i\j_1\j_5
 +i\j_2\j_4
 -i\j_2\j_7
 -i\j_3\j_8
 +i\j_4\j_6
 \\[-1pt]
\omit&~~
 +\f_2F_2
 +\f_3F_1
 -\f_5F_3
 +\f_1\dt\f_2
 -\f_1\dt\f_4
 +\f_3\dt\f_5
\\
\bottomrule
  \end{array}
$$\vspace{-7mm}
  \caption{Candidate Lagrangian summands $i\C6{Q_3}\C3{Q_1}\,Z(\f_*)$}
  \label{t:ZE02}
\end{table}
\begin{table}[htpb]
\small$$
  \begin{array}{@{} r@{:~}l @{}}
    \toprule
\fc12(\f_1\f_1{-}\f_2\f_2)
&
 -i\j_1\j_4
 -i\j_2\j_3
 +\f_1F_3
 -\f_2F_1
 +\f_1\dt\f_3
\\
\f_1\f_2
&
 -i\j_1\j_2
 -i\j_3\j_4
 +\f_1F_1
 +\f_2F_3
 +\f_2\dt\f_3
\\
\f_1\f_3{+}\f_2\f_5
&
 -i\j_1\j_8
 -i\j_2\j_6
 -i\j_3\j_5
 -i\j_4\j_7
 +\f_1F_2
 +\f_3F_3
 +\f_5F_1
 +\f_2\dt\f_4
\\
\f_1\f_4{+}\f_3\f_5
&
 -i\j_5\j_7
 +i\j_6\j_8
 +\f_4F_3
 +\f_5F_2
\\
\f_1\f_5{-}\f_2\f_3
&
 +i\j_1\j_3
 -i\j_1\j_5
 -i\j_2\j_4
 +i\j_2\j_7
 +i\j_3\j_8
 -i\j_4\j_6
 \\[-1pt]
\omit&~~
 -\f_2F_2
 -\f_3F_1
 +\f_5F_3
 -\f_1\dt\f_2
 +\f_1\dt\f_4
 -\f_3\dt\f_5
\\
\f_1\f_5{-}\f_3\f_4
&
 +i\j_5\j_8
 +i\j_6\j_7
 -\f_4F_2
 +\f_5F_3
\\
\f_4\f_4{-}\f_5\f_5
&
~\C5{-2i\vdt(\f_4\f_5)}
\\
\f_4\f_5
&
~\C5{\frc{i}2\vdt(\f_4^{~2}{-}\f_5^{~2})}
\\
~~\fc12(\f_1\f_1{-}\f_3\f_3{+}\f_5\f_5){-}2 \f_2\f_4
&
 -i\j_1\j_7
 -i\j_2\j_5
 +i\j_3\j_6
 +i\j_4\j_8
 -i\j_5\j_6
 -i\j_7\j_8
 \\[-1pt]
\omit&~~
 +\f_1F_3
 -\f_3F_2
 -\f_4F_1
 +\f_2\dt\f_5
 -\f_4\dt\f_5
\\
\bottomrule
  \end{array}
$$\vspace{-7mm}
  \caption{Candidate Lagrangian summands $i\C6{Q_3}\C1{Q_2}\,Z(\f_*)$}
  \label{t:ZE03}
\end{table}

\subsection{Supersymmetry Extensions}
\label{a:SuSyX}
The $\C8{Q_4}$- and $\C8{\Tw{Q}_4}$-transform of the $N\,{=}\,3$-supersymmetric Lagrangian terms in Tables~\ref{t:KE} and~\ref{t:SZ} are computed straightforwardly and listed below in turn.

\paragraph{The Chiral-Chiral Extension:}
Using the chiral-chiral $\C8{Q_4}$-supersymmetry transformation rules\eq{e:GKY4}, we compute:
\begin{alignat}9
 \C8{Q_4}K^1
 &=\C5{-\vdt\big(F_1\j_3 +(F_3{+}\dt\f_3)\j_1 -\dt\f_1\j_4 -\dt\f_2\j_2\big)}
 \label{e:Q4K1}
\\[2pt]
 \C8{Q_4}K^2
 &= \C5{\vdt\big(F_2\j_7 -F_3\j_8 +\dt\f_4\j_6 +\dt\f_5\j_5\big)}
 \label{e:Q4K2}
\\[2pt]
 \C8{Q_4}K^3
 &=\C5{-\vdt\big(F_2\j_3 +F_3\j_2 +\dt\f_4\j_1 -\dt\f_5\j_4\big)}
 \label{e:Q4K3}
\\[2pt]
 \C8{Q_4}K^4
 &=\C5{\vdt\big(F_2\j_2 -F_3\j_3 -\dt\f_4\j_4 -\dt\f_5\j_1\big)}
 \label{e:Q4K4}
\\[2pt]
 \C8{Q_4}K^5
 &=\C5{\vdt\big(F_2\j_4 -F_3\j_1 +\dt\f_4\j_2 +\dt\f_5\j_3\big)}
 \label{e:Q4K5}
\\[2pt]
 \C8{Q_4}K^6
 &=\C5{-\vdt\big(F_2\j_1 +F_3\j_4 -\dt\f_4\j_3 +\dt\f_5\j_2\big)}
 \label{e:Q4K6}
\\[6pt]
 \C8{Q_4}Z^1
 &= \C5{\vd_t\big(\f_4\j_8-\f_5\j_7\big)}
 \label{e:Q4Z1}
\\[2pt]
 \C8{Q_4}Z^2
 &=\C5{-\vdt\big(\f_5\j_8 +\f_4\j_7\big)}
 \label{e:Q4Z2}
\\[2pt]
 \C8{Q_4}Z^3
 &= \C5{\vdt\big(\f_2\j_1 +\f_1\j_3\big)}
 \label{e:Q4Z3}
\\[2pt]
 \C8{Q_4}Z^4
 &= \C5{\vdt\big(\f_1\j_1 -\f_2\j_3\big)}
 \label{e:Q4Z4}
\\[2pt]
 \C8{Q_4}Z^5
 &=\C5{-\vdt\big(\dt\f_1\j_7 +\dt\f_2\j_6 -\dt\f_3\j_8 -\dt\f_5\j_3\big)}
 \label{e:Q4Z5}
\\[2pt]
 \C8{Q_4}Z^6
 &= \C5{\vdt\big(\f_1\j_8 -\f_2\j_5 +\f_3\j_7 -\f_4\j_3 +\f_4\j_5\big)}
 \label{e:Q4Z6}
\\[2pt]
 \C8{Q_4}Z^7
 &=-\C5{\vdt\big(\f_1\j_2 -\f_1\j_6 +\f_2\j_7 -\f_3\j_3 +\f_3\j_5 +\f_5\j_8\big)}
 \label{e:Q4Z7}
\end{alignat}
Thus, all thirteen Lagrangian terms are automatically also $\C8{Q_4}$-supersymmetric.

\paragraph{The Chiral-Twisted-Chiral Extension:}
Using the chiral-twisted-chiral $\C8{\Tw{Q}_4}$-supersymmetry transformation rules\eq{e:GKY4-}, we compute:
\begin{alignat}9
 \C8{\Tw{Q}_4}K^1
 &=\C5{-\vdt\big(F_1\j_3 +F_3\j_1
                  -\dt\f_1 \j_4 -\dt\f_2 \j_2 +\dt\f_3 \j_1 \big)}
 \label{e:Q4tK1}
\\[2pt]
 \C8{\Tw{Q}_4}K^2
&=\C5{-\vdt\big(F_2\j_7 -F_3\j_8 +\dt\f_4\j_6 +\dt\f_5\j_5\big)}
 \label{e:Q4tK2}
\\[2pt]
 \C8{\Tw{Q}_4}K^3
 &=
 -2\big(F_1 \dt\j_7
       +F_2 \dt\j_3
       +F_3(\dt\j_2+\dt\j_6)
       +\dt\f_1\dt\j_5
       +\dt\f_2 \dt\j_8
       +\dt\f_3\dt\j_6
       +\dt\f_4 \dt\j_1
       -\dt\f_5 \dt\j_4\big)\nn\\
 &\qquad
  \C5{+\vdt\big( F_2 \j_3 +F_3 \j_2 +\dt\f_4 \j_1 -\dt\f_5 \j_4\big)}
 \label{e:Q4tK3}
\\[2pt]
 \C8{\Tw{Q}_4}K^4
 &=
 +2\big( F_1 \dt\j_8 
        +F_2 \dt\j_2
        -F_3 (\dt\j_3{+}\dt\j_5)
        +\dt\f_1 \dt\j_6
        -\dt\f_2 \dt\j_7
        -\dt\f_3 \dt\j_5
        -\dt\f_4 \dt\j_4
        -\dt\f_5 \dt\j_1 \big)
 \nn\\
 &\qquad
  \C5{-\vdt\big( F_2\j_2 -F_3\j_3 -\dt\f_4\j_4 -\dt\f_5\j_1 \big)}
 \label{e:Q4tK4}
\\[2pt]
 \C8{\Tw{Q}_4}K^5
 &=
 +2\big( F_1 \dt\j_5
        +F_2 \dt\j_4
        -F_3(\dt\j_1{-}\dt\j_8)
        -\dt\f_1 \dt\j_7
        -\dt\f_2 \dt\j_6
        +\dt\f_3 \dt\j_8
        +\dt\f_4 \dt\j_2
        +\dt\f_5 \dt\j_3 \big)
 \nn\\
 &\qquad\C5{-\vdt\big( F_2 \j_4 -F_3 \j_1 +\dt\f_4 \j_2 +\dt\f_5 \j_3 \big)}
 \label{e:Q4tK5}
\\[2pt]
 \C8{\Tw{Q}_4}K^6
 &=
 +2\big( F_1 \dt\j_6
        -F_2 \dt\j_1
        -F_3 \dt\j_4
        -F_3 \dt\j_7
        -\dt\f_1 \dt\j_8
        +\dt\f_2 \dt\j_5
        -\dt\f_3 \dt\j_7
        +\dt\f_4 \dt\j_3
        -\dt\f_5 \dt\j_2 \big)
 \nn\\
 &\qquad\C5{+\vdt\big( F_2\j_1 +F_3\j_4 -\dt\f_4\j_3 +\dt\f_5\j_2 \big)}
 \label{e:Q4tK6}
\\[6pt]
 \C8{\Tw{Q}_4}Z^1
 &=\C5{-\vdt\big(\f_4\j_8-\f_5\j_7\big)}
 \label{e:Q4tZ1}
\\[2pt]
 \C8{\Tw{Q}_4}Z^2
 &=\C5{+\vdt\big(\f_4\j_7+\f_5\j_8\big)}
 \label{e:Q4tZ2}
\\[2pt]
 \C8{\Tw{Q}_4}Z^3
 &=\C5{+\vdt\big(\f_1\j_3+\f_2\j_1\big)}
 \label{e:Q4tZ3}
\\[2pt]
 \C8{\Tw{Q}_4}Z^4
 &=\C5{+\vdt\big(\f_1\j_1-\f_2\j_3\big)}
 \label{e:Q4tZ4}
\\[2pt]
 \C8{\Tw{Q}_4}Z^5
 &=
 +2\big( F_1 \j_5
        -F_2 \j_4
        +F_3(\j_1{+}\j_8)
        -\dt\f_1 \j_7
        -\dt\f_2 \j_6
        +\dt\f_3 \j_8
        -\dt\f_4 \j_2
        -\dt\f_5 \j_3 \big)
 \nn\\
 &\qquad\C5{+\vdt\big( +\f_1\j_7 +\f_2\j_6 -\f_3\j_8 +\f_5\j_3 \big)}
 \label{e:Q4tZ5}
\\[2pt]
 \C8{\Tw{Q}_4}Z^6
&=
 -2 \big( F_1 \j_6
         +F_2 \j_1
         +F_3(\j_4{-}\j_7)
         -\dt\f_1 \j_8
         +\dt\f_2 \j_5
         -\dt\f_3 \j_7
         -\dt\f_4 \j_3
         +\dt\f_5 \j_2 \big)
 \nn\\
 &\qquad\C5{-\vdt\big( \f_1\j_8 -\f_2\j_5 +\f_3\j_7 +\f_4\j_3 +\f_4\j_5 \big)}
 \label{e:Q4tZ6}
\\[2pt]
 \C8{\Tw{Q}_4}Z^7
&=
 +2 \big( F_1 \j_8
         -F_2 \j_2
         +F_3 \j_3
         -F_3 \j_5
         +\dt\f_1 \j_6
         -\dt\f_2 \j_7
         +\dt\f_3 \j_5
         +\dt\f_4 \j_4
         +\dt\f_5 \j_1 \big)
\nn\\
&\quad\C5{-\vdt\big(
  \f_1\j_2
 +\f_1\j_6
 -\f_2\j_7
 -\f_3\j_3
 -\f_3\j_5
 -\f_5\j_8
 \big)}
 \label{e:Q4tZ7}
\end{alignat}
Thus, only five Lagrangian terms ( $K^1,K^2,Z^1,Z^2,Z^3$ and $Z^4$) are also $\C8{\Tw{Q}_4}$-supersymmetric.

\end{document}